\documentclass[journal]{IEEEtran}
\ifCLASSINFOpdf
\else
\fi

\usepackage{amsmath,amsthm,amssymb,amsfonts}
\usepackage{mathtools}
\usepackage{mathrsfs}
\usepackage{algorithmic}
\usepackage[dvipdfmx]{graphicx}
\usepackage{subfigure}
\usepackage[ruled,linesnumbered]{algorithm2e}
\usepackage{array}
\usepackage{textcomp}
\usepackage{stfloats}
\usepackage{tabularx}
\usepackage{lipsum}
\usepackage{cite}
\usepackage{bm}
\allowdisplaybreaks[4]
\usepackage{color}
\usepackage{cases}
\newtheorem{theorem}{Theorem}
\newtheorem{lemma}{Lemma}

\makeatletter
\renewcommand{\maketag@@@}[1]{\hbox{\m@th\normalsize\normalfont#1}}%
\makeatother

\usepackage[hidelinks]{hyperref}
\begin{document}

\title{Covert Communication in Hybrid Microwave/mmWave A2G Systems with Transmission Mode Selection}

\author{Wenhao Zhang,~\IEEEmembership{Student Member,~IEEE,}
		Ji He,~\IEEEmembership{Member,~IEEE,}
		Yulong Shen,~\IEEEmembership{Member,~IEEE,}
		Xiaohong Jiang,~\IEEEmembership{Senior Member,~IEEE,}
 \thanks{W. Zhang and X. Jiang are with the School of System Information Science, Future University Hakodate, Hakodate, 041-8655 Japan (e-mail:hizhangwenhao@gmail.com; jiang@fun.ac.jp).}
 \thanks{J. He and Y. Shen are with the School of Computer Science and Technology, Xidian University, Xi'an, 710071 China (e-mail: garyhej1991@gmail.com;ylshen@mail.xidian.edu.cn).}
}


\maketitle

\begin{abstract}
This paper investigates the covert communication in an air-to-ground (A2G) system, where a UAV (Alice) can adopt the omnidirectional microwave (OM) or directional mmWave (DM) transmission mode to transmit covert data to a ground user (Bob) while suffering from the detection of an adversary (Willie). For both the OM and DM modes, we first conduct theoretical analysis to reveal the inherent relationship between the transmit rate/transmit power and basic covert performance metrics in terms of detection error probability (DEP), effective covert rate (ECR), and covert Shannon capacity (CSC). To facilitate the transmission mode selection at Alice, we then explore the optimization of transmit rate and transmit power for ECR/CSC maximization under the OM and DM modes, and further propose a hybrid OM/DM transmission mode which allows the UAV to adaptively select between the OM and DM modes to achieve the maximum ECR and CSC at a given location of UAV. Finally, extensive numerical results are provided to illustrate the covert performances of the concerned A2G system under different transmission modes, and demonstrate that the hybrid OM/DM transmission mode outperforms the pure OM or DM mode in terms of covert performance.
\end{abstract}



\section{Introduction}

With the ever-increasing applications of wireless networks, an unprecedented amount of private and sensitive information is transferred over the wireless medium. However, the broadcast characteristic of wireless channels will lead to serious security risks. Although cryptography is the well-known security method, it has been demonstrated that even the most robust cryptography is at risk of being brutally cracked by quantum computing \cite{chen2016report}. As an alternative approach, physical layer security can achieve information-theoretic security by exploiting the inherent physical layer randomness of wireless channels \cite{zhou2013physical}. Recently, covert communication has been regarded as a new security paradigm to hide the communication process itself to achieve a stronger level of security \cite{yan2019low}.

Bash et al. first derived the information theory limitation called square root law (SRL) for covert communication \cite{bash2013limits}, pointing out that under the additive white Gaussian noise (AWGN) channel, Alice can only reliably and covertly send $\mathcal{O}(\sqrt{n})$ bits to Bob over $n$ channel uses. Subsequently, the fundamental limit of covert communication has been further studied under various wireless channels, including binary symmetric channels \cite{che2013reliable}, discrete memory-less channels \cite{wang2016fundamental}, and multiple input multiple output (MIMO) AWGN channels \cite{abdelaziz2017fundamental}. Following this line, extensive research works have been dedicated to the study of covert communication in various network scenarios. According to the movability of network nodes, these available works can be classified as studies for network scenarios with fixed nodes and studies for network scenarios with mobile nodes.

For the network scenarios with fixed nodes, the authors in \cite{lee2015achieving}  first proved that a nonzero covert rate is possible over the AWGN channels and Rayleigh channels regardless of the finite or infinite channel use. Based on the results, the optimal power adaptation schemes were designed in \cite{li2020optimal}. Considering the channel uncertainty, the work \cite{shahzad2020covert} reveals the fundamental difference in the design of covert transmission schemes between the case of quasi-static fading channel and the case of non-fading AWGN channel \cite{shahzad2017covert}. The jamming-assisted covert communication was also explored in different network scenarios, including the ones with full-duplex jamming receiver \cite{shahzad2018achieving}, the ones with Poisson point distributed jammer \cite{soltani2018covert}, the ones with randomly located warden \cite{zheng2019multi}, the ones with two-hop relaying \cite{bai2022covert}, and the ones with intelligent reflecting surface (IRS) technology \cite{LvWLDAC22,MamaghaniH22}. Notice that the millimeter-wave (mmWave) technology can bring high bandwidth, the covert communications with mmWave has been examined in some recent works. The authors in \cite{jamali2021covert} first considered a mmWave communication system with dual independent antenna arrays, and proposed a covert transmission scheme where one antenna array is used to form a beam towards Bob for data transmission and the other is used for jamming. The covert rate maximization in mmWave system was investigated in \cite{zhang2021joint} through the joint optimization of beam training duration, training power and transmission power. A multi-user beam training strategy was proposed in \cite{zhang2022multi} to maximize the effective covert throughput. Besides, the authors in \cite{wang2021covert} considered a full-duplex covert mmWave communications scenario, and explored the maximization of achievable covert rate by jointly optimizing the beamforming, transmit power, and jamming power.

For the mobile node network scenario, the impact of node mobility on the covert performance in the mobile Ad Hoc network was first studied in \cite{im2020mobility}. Considering the UAV-assisted covert communications, the authors in \cite{yang2021mode} discussed an adaptive switching scheme between half-duplex and full-duplex modes to enhance the covert capacity. Subsequently, for improving the covert performance of air-to-ground (A2G) systems, the authors in \cite{zhou2019joint} focused on the corresponding optimization of UAV’s trajectory and transmission power and the authors in \cite{zhou2021three} explored the corresponding optimization of UAV’S placement. Then, the joint optimization of transmission power and jamming power of full-duplex UAV was investigated in \cite{chen2021uav} to enhance the covert performance in the two-hop UAV relayed covert communication system. Regarding the joint optimization of trajectory and power for the covert communication with UAV jamming, different optimization algorithms have been developed, e.g., penalty successive convex approximation algorithm \cite{zhou2021uav}, model-driven generative adversarial network algorithm \cite{li2021md}, and geometric method \cite{rao2022optimal}. From the monitor’s point of view, the authors in \cite{hu2019optimal} developed a novel detection strategy which adopts multiple antennas with beam sweeping to confront the UAV covert communication. In \cite{wang2019secrecy}, the authors further discussed the trade-off between secrecy/covertness and efficiency of the multi-hop transmission with a UAV warden, and studied the related throughput optimization by carefully designing the parameters of the network, including the coding rates, transmit power, and required number of hops. It is expected that the mmWave and Terahertz communication can meet the demand of high data rate applications for forthcoming beyond-5G (B5G) wireless communications, the authors in \cite{zhang2020optimized} and \cite{mamaghani2022aerial} conducted some initial studies to investigate the covert performance of the UAV-aided mmWave wireless system and Terahertz wireless system, respectively.

The aforementioned works help us understand the covert communication with either microwave or mmWave technologies. Notice that these two technologies have their strength and weakness, so some recent work explored the dual-frequency antenna design to support both bands of communications \cite{xiang2017flexible,zhang2018dual,Li2018} as well the related issue of network performance analysis \cite{vuppala2016analysis,semiari2017joint}. These works indicate that network performance can be significantly enhanced by supporting both microwave and mmWave in the communication system. Thus, a natural question that arises is how to take the full advantages of microwave and mmWave to achieve a more efficient covert communication. To this end, we explore covert communication in an A2G wireless system where UAV is equipped with a dual-band antenna and proposed a hybrid omnidirectional microwave (OM) and directional mmWave (DM) transmission mode to improve covert performance of the system. The main contributions of this paper are summarized as follows.

\begin{itemize}
\item We first develop theoretical frameworks for covert performance analysis of an A2G system under both the OM and DM transmission modes, in terms of the detection error probability (DEP), the effective covert rate (ECR) and the covert Shannon capacity (CSC). Specially, the uncertainty of line-of-sight (LoS) and non-line-of-sight (NLoS) of the A2G channel due to the movement of UAV is fully considered in the analysis. These frameworks reveal the inherent relationship between the key system parameters (i.e., UAV’s location, transmission rate, transmission power) and the DEP, ECR and CSC.
\item We then explore the covert performance optimization of the concerned A2G wireless system under both the OM and DM transmission modes. Specifically, for a given position of UAV, we determine the optimal covert transmission power and transmission rate to achieve the maximal ECR and the maximal CSC while satisfying the DEP constraint. Based on the optimization results, we further propose a hybrid OM/DM transmission mode which allows the UAV to adaptively select between the OM and DM transmission modes at different locations to achieve the optimal covert transmission performance.
\item Finally, extensive numerical results are provided to illustrate the effects of the key system parameters on the overall covert performance and also to demonstrate that the proposed hybrid OM/DM transmission mode can achieve an enhanced covert performance than the pure OM or DM transmission mode in the considered A2G system.
\end{itemize}

The remainder of the paper is organized as follows. Section \ref{sec2} introduces the system model and preliminaries. Section \ref{sec3} presents the concerned transmission modes and detection strategy at Willie. The covert performance modeling of the OM and DM transmission modes is provided in \ref{TA_Performance} and Section \ref{MA_Performance}, respectively. Section \ref{sec_results} shows the covert performance optimization and the hybrid OM/DM transmission mode. The numerical results are provided in Section  \ref{sec_num_results}. Finally, we conclude in Section \ref{sec_conclusion}. The notations used in this paper are summarized in Table \ref{notations table}.

\renewcommand{\arraystretch}{1.4}
\begin{table}[th]
	\centering
	\caption{List of Notations}
	\label{notations table}
	\begin{tabular}{|m{0.13\textwidth}|m{0.31\textwidth}|}
		\hline
		\textbf{Notation} & \textbf{Description} \\
		\hline
		\hline
		$P_{a}$, $d_{ij}$, $h_{i}$ & Transmission power of Alice, Distance between nodes $i$ and $j$, height of node $i$ \\
		\hline
		$d_{aw}^{\min}$, $d_{aw}^{\max}$ & Minimum distance and maximum distance between Alice and Willie\\
		\hline
		$\mathbb{A}\in \{M,S\}$ & Main lobe and side lobe, respectively\\
		\hline
		$\mathbb{B}\in \{L,N\}$ & LoS channel and NLoS channel, respectively\\
		\hline
		$\mathbb{C}\in \{o,d\}$ & OM transmission mode and DM transmission mode, respectively\\
		\hline
		$G_{ij}, G_{i}^{M}, G_{i}^{S}$ & Total antenna gains of link $i\to j$, main lobe gain, side lobe gain\\
		\hline
		$P_{ij}^{L}, P_{ij}^{N}$ & Probability that the channel mode of link $i\to j$ is LoS link and NLoS link, respectively\\
		\hline
		$H_{ij}^{\mathbb{C},\mathbb{B}}, h_{ij}^{\mathbb{C},\mathbb{B}},L_{ij}^{\mathbb{C},\mathbb{B}}$ & Channel coefficient, small scaling fading coefficient and path loss of link i-j, respectively\\
		\hline
		$\alpha_{\mathbb{B}}, \beta _{\mathbb{B}}, S_{\mathbb{B}}, k_{\mathbb{B}}$& Path-loss exponent, path-loss coefficient, Nakagami fading parameter, Rician factor of $\mathbb{B}$, respectively\\
		\hline
		$\theta_{i}$, $\theta_{i}^{h}$, $\theta_{i}^{e/d}$, $\varphi_{i}$ & Elevation angle of UAV relative to node $i$, half-power beamwidth for the azimuth orientation in the horizontal direction of node $i$, half-power beamwidth for the elevation/depression angles of node $i$, azimuth orientation in the vertical direction of node $i$ relative to UAV, respectively\\
		\hline
		$\sigma _{b}^{2}$, $\sigma _{w}^{2}$, $\hat{\sigma} _{n}^{2}$ & Noise power of Bob, noise power of Willie, nominal noise power of Willie\\
		\hline
		$\epsilon$, $\gamma_{th}$, $R_{b}$ & Covertness constraint, outage probability threshold, target rate, respectively.\\
        \hline
	\end{tabular}
\end{table}


\section{System Model and Preliminaries} \label{sec2}

\subsection{System Model}

\begin{figure} [t]
    \centering
    \includegraphics[width=0.45\textwidth]{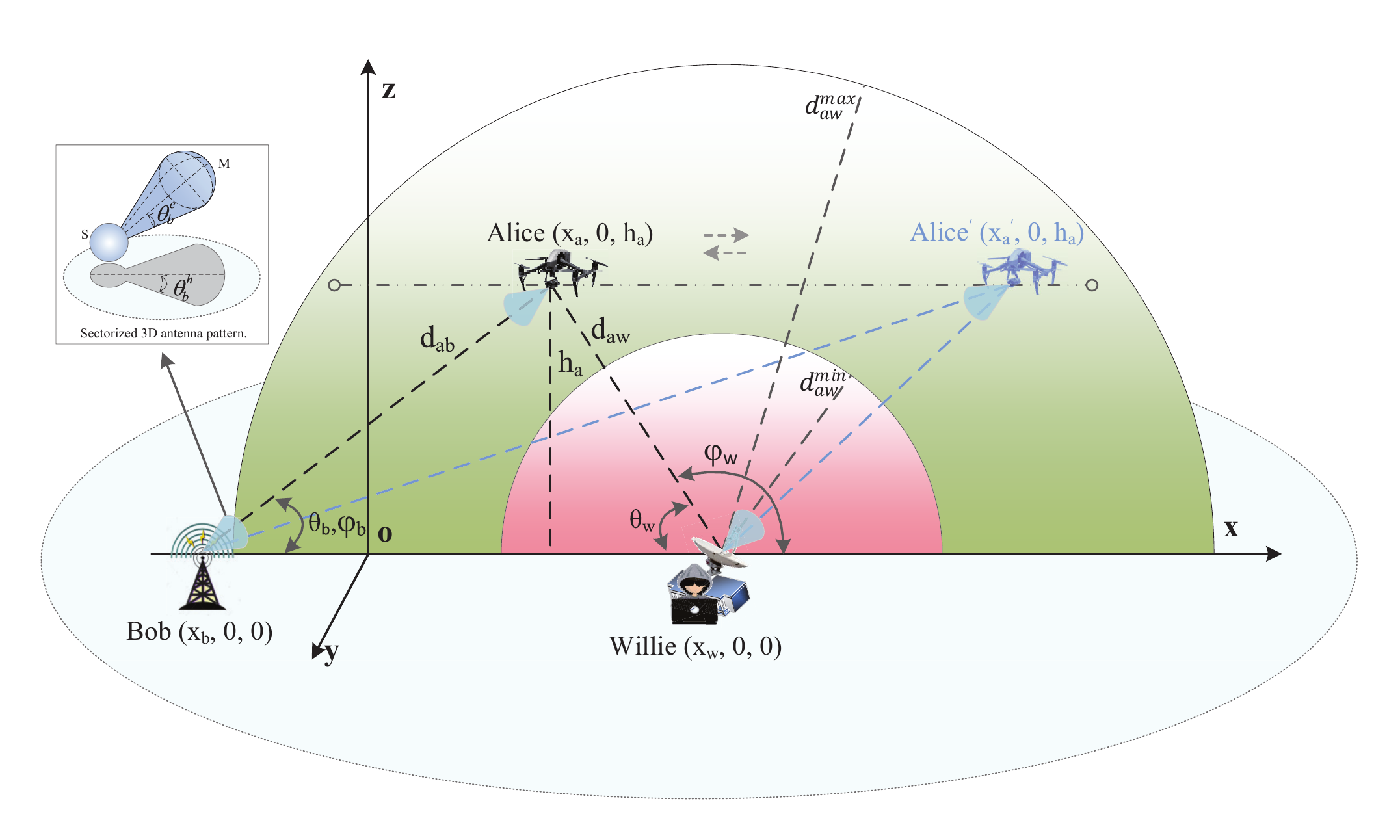}
    \caption{Illustration of the concerned A2G system model.}
    \label{fig:system_model}
    \vspace{-0.3cm}
\end{figure}

We consider an A2G covert communication scenario shown in Fig.~\ref{fig:system_model}, where an UAV (Alice, $a$) executes surveillance task (e.g., imaging) over a geographic area and urgently transmits critical information back to the intended user (Bob, $b$), but suffering from the detection of a warden (Willie, $w$). Alice (resp., Bob and Willie) is configured with the integrated antennas of the OM and DM as \cite{xiang2017flexible,zhang2018dual,Li2018}, such that it can send (resp., receive) both microwave signals and mmWave signals. Alice employs a fixed transmission power $P_a$ subject to maximum power constraint $P_{max}$. The UAV cyclically flies with a horizontal trajectory around the surveillance area at a fixed height $h_{a}$. We let $d_{ij}$ denote the distance between node $i$ and node $j$ for $i ,j \in\{a,b,w\}$. To prevent the UAV from being detected by vision or radar, there exists a minimum distance constraint $d_{aw}^{\min}$ between Alice and Willie. Besides, due to the finite resolution of the UAV's imaging equipment, a maximum distance constraint $d_{aw}^{\max}$ is essential.

\subsection{Antenna Gains}

Because UAV moves in 3D space, we adopt the 3D sectorized antenna pattern as \cite{venugopal2016device,zhu2018secrecy}, which is essential to model the links from air to ground. As shown in Fig.~\ref{fig:system_model}, $\theta_{i}^{h}$ is the half-power beamwidth for the azimuth orientation in the horizontal direction, $\theta_{i}^{e}$ (resp. $\theta_{i}^{d}$) is the half-power beamwidth for the elevation (resp. depression) angles in the ground or air node $i$. Considering the antenna gain $G_{i}^{M}$ within the half-power beamwidth (i.e., main lobe) and gain $G_{i}^{S}$ outside the half-power beamwidth (i.e., side lobe), the total directivity gain of DM for link $i\to j$ is determined as $G_{ij}=G_{i}^{\mathbb{A}}G_{j}^{\mathbb{A}}$, where $\mathbb{A}\in \{M,S\}$. To improve the signal quality at Bob, Alice dynamically adjusts the steering orientation of the antenna array toward Bob, such that she can maximize antenna array gain $G_{a}^{M}$. Thus, the directional antenna gain of Alice$\to$ Bob link is determined as $G_{ab}=G_{a}^{M}G_{b}^{\mathbb{A}}$. Note that in the process of Alice's movement, Willie may be located in the half-power beamwidth direction of the antenna array of Alice$\to$ Bob link (i.e., $\pi-\varphi_{w}-\varphi_{b} \leq \theta_{a}^{d}$) and the antenna array gain of Alice for Willie is the main lobe gain $G_{a}^{M}$, otherwise, the side lobe gain $G_{a}^{S}$. Thus, the antenna gain of Alice$\to$ Willie link is determined as
\begin{equation} \label{neq1_gaw}
\begin{aligned}
 G_{aw}=G_{w}^{\mathbb{A}} \times\begin{cases}
  G_{a}^{M},& \pi-\varphi_{w}-\varphi_{b} \leq \theta_{a}^{d},\\
 G_{a}^{S},& otherwise.
\end{cases}
\end{aligned}
\end{equation}
Considering an actual situation, we adopt a uniform planar square array (UPA) to model the sectorized pattern at node $j$ (here, $j\in\{b,w\}$). The directional antenna gain and the associated probability can be approximated as
\begin{equation} \label{neq1}
\begin{aligned}
 G_{j}^{\mathbb{A}}= \begin{cases}
 G_{j}^{M}=\mathcal{N}_j,& \!P_{j}^{M}\!=  \frac{\theta_{j}^{h}}{2\pi}\times  \frac{\theta_{j}^{e}}{\pi},\\
 G_{j}^{S}=\frac{\sqrt{\mathcal{N}_j}-\frac{\sqrt{3}}{2\pi}\mathcal{N}_j\sin{\left( \frac{3\pi}{2\sqrt{\mathcal{N}_j}}\right)}}{\sqrt{\mathcal{N}_j}-\frac{\sqrt{3}}{2\pi}\sin{\left( \frac{3\pi}{2\sqrt{\mathcal{N}_j}}\right)}} ,& P_{j}^{S}= 1- \frac{\theta_{j}^{h}}{2\pi}\times \! \frac{\theta_{j}^{e}}{\pi},
\end{cases}
\end{aligned}
\end{equation}
where $\mathcal{N}_j$ denotes the number of antenna elements at the node $j$, $P_{j}^{\mathbb{A}}$ denotes the associated probability.


\subsection{Propagation Model}

To characterize the A2G wireless channel feature, the channel from UAV to Bob (Willie) is modeled as the combination of NLoS and LoS channels. Specifically, we let $P_{aj}^{\mathbb{B}}$ ($\mathbb{B}\in \{L,N\}, j \in \{b, w\}$) denotes the probability the channel from the UAV to node $j$ is the LoS or NLoS, where $P_{aj}^{L}=(1+\sigma\exp{(-f[\theta_{j}-\sigma])})^{-1}$ and $P_{aj}^{N}\triangleq 1-P_{aj}^{L}$ as \cite{andrews2016modeling,al2014optimal}, where $\theta_{j}=\frac{180}{\pi}\arcsin(\frac{h_{a}}{d_{aj}})$ is the degree of the elevation angle for node $j$ relative to UAV, $d_{aj} =\sqrt{(x_{a}-x_{j})^2+h^2_a}$ is the distance between UAV and node $j$, $\sigma$ and $f$ are the S-curve parameters that depend on the communication environment. In this work, we consider both the large scale fading and small scale fading. Thus, the channel coefficient for link $a\to j$ is denoted as $H_{aj}^{\mathbb{C},\mathbb{B}}=h^{\mathbb{C},\mathbb{B}}_{aj}\sqrt{L_{aj}^{\mathbb{C},\mathbb{B}}}$, where $\mathbb{C} \in \{ o,m\}$ denotes adopting the OM or DM, $h^{\mathbb{C},\mathbb{B}}_{aj}$ and $L_{aj}^{{\mathbb{C},\mathbb{B}}}$ denote the channel fading coefficient and path loss, respectively. We have $L_{aj}^{\mathbb{C},\mathbb{B}}=\beta^{\mathbb{C}}_{\mathbb{B}}d_{aj}^{-\alpha^{\mathbb{C}}_{\mathbb{B}}}$, where $\beta^{\mathbb{C}}_{\mathbb{B}}$ and $\alpha^{\mathbb{C}}_{\mathbb{B}}$ are the constant coefficient and path-loss exponent, respectively.

\subsubsection{Small scaling fading of Microwave channel}
Due to the possible combination of LoS and multipath scatterers of NLoS, the microwave channels are modeled as Rician fading to characterize the propagation effect. Thus, the fading coefficient $h^{o,\mathbb{B}}_{ij}$ can be described as 
\begin{equation} \label{Eq:rician_h}
\begin{aligned}
h^{o,\mathbb{B}}_{ij}=\sqrt{\frac{k_{\mathbb{B}}}{k_{\mathbb{B}}+1}}\hat{h}_{ij}+\sqrt{\frac{1}{k_{\mathbb{B}}+1}}\tilde{h}_{ij},
\end{aligned}
\end{equation}
where $|\hat{h}_{ij}|=1$ and $\tilde{h}_{ij}\sim\mathcal{CN}(0,1)$ denote the components of LoS channel and NLoS Rayleigh fading, respectively, $k_{\mathbb{B}}$ denotes the Rician factor which is determined as \cite{shimamoto2006channel,you20193d}
\begin{equation} \label{Eq:rician_k}
\begin{aligned}
k_{\mathbb{B}}(\theta _{j})=\begin{cases} \eta_{1} \exp{(\eta_{2}\theta _{j})},&
 LoS, \\0 ,& NLoS, \end{cases}
\end{aligned}
\end{equation}
$\eta_{1}$ and $\eta_{2}$ are constant coefficients with $\eta_{1}=k_0$ and $\eta_{2}=\frac{2}{\pi} \ln{(\frac{k_\frac{\pi}{2}}{k_0})}$, $k_0=k_{\mathbb{B}}(0)$ and $k_\frac{\pi}{2}=k_{\mathbb{B}}(\frac{\pi}{2})$ are depending on the specific environment. Thus, the fading coefficient $|h^{o,\mathbb{B}}_{ij}|^{2}$ follows a non-central chi-squared distribution with two degrees of freedom and $\mathbb{E}[|h^{o,\mathbb{B}}_{ij}|^{2}]=1$, the probability density function (PDF) of $|h^{o,\mathbb{B}}_{ij}|^{2}$ is expressed as \cite{simon2005digital}
\begin{equation} \label{Eq:rician_pdf}
\begin{aligned}
f_{|h^{o,\mathbb{B}}_{ab}|^{2}}(x)\!=\!(k_{\mathbb{B}}\!+\!1)e^{-k_{\mathbb{B}}-(k_{\mathbb{B}}+1)x}I_{0}(2\sqrt{k_{\mathbb{B}}(k_{\mathbb{B}}\!+\!1)x}),
\end{aligned}
\end{equation}
where $I_{0}(\cdot)$ is the modified Bessel function of the first kind and zeroth order, and the corresponding CDF of $|h^{o,\mathbb{B}}_{ij}|^{2}$ can be expressed as \cite{simon2005digital,you20193d}
\begin{equation} \label{Eq:rician_cdf}
\begin{aligned}
F_{|h^{o,\mathbb{B}}_{ij}|^{2}}(x)=1\!-Q_{1}\!\!\left (\! \sqrt{2k_{\mathbb{B}}},\sqrt{2(k_{\mathbb{B}}+1)x} \right ),
\end{aligned}
\end{equation}
where $Q_{1}(a,b)$ is the standard Marcum-Q function which is
\begin{equation} \label{Eq:rician_q}
\begin{aligned}
Q_{1}(a,b)=\int_{b}^{\infty}xI_{0}(ax)\exp{\left(-\frac{x^{2}+a^{2}}{2}\right)}dx.
\end{aligned}
\end{equation}
$I_{0}(\cdot)$ is mentioned above.

\subsubsection{Small scaling fading of mmWave channel}
mmWave channels are characterized by the Nakagami-m fading with shape parameter $S_{\mathbb{B}} \geq 1/2$ and scale parameter $\Omega =\mathbb{E}[|h^{d,\mathbb{B}}_{ij}|^{2}]=1$ as \cite{bai2014coverage,andrews2016modeling,zhu2018secrecy}. Thus, the fading coefficient $|{h}^{d,\mathbb{B}}_{ij}|^{2}$ follows a normalized gamma distribution with shape and scale parameters of $S_{\mathbb{B}}$ and $1/S_{\mathbb{B}}$. Therefore, the PDF is given by \cite{simon2005digital,jamali2021covert}
\begin{equation} \label{eq_mmwave_pdf}
\begin{aligned}
f_{|h^{d,\mathcal{B}}_{ab}|^{2}}(x)=S_{\mathcal{B}}^{S_{\mathcal{B}}}x^{S_{\mathcal{B}}-1}e^{-S_{\mathcal{B}}x}/\Gamma(S_{\mathcal{B}}).
\end{aligned}
\end{equation}
According to \cite[Lemma 6]{bai2014coverage} and Alzer's lemma \cite{alzer1997some}, the CDF of a normalized gamma RV is tightly approximated by $[1-\exp{(-\xi _{\mathbb{B}}x)}]^{S_{\mathbb{B} }}$ where $\xi _{\mathbb{B}}=S_{\mathbb{B} }(S_{\mathbb{B} }!)^{-1/S_{\mathbb{B}}}$. By applying the binomial theorem assuming $S_{\mathbb{B} }$ is an integer \cite{bai2014coverage}, the CDF of $|h^{d,\mathbb{B}}_{ij}|^{2}$ is given by
\begin{equation} \label{eq10}
\begin{aligned}
F_{|h^{d,\mathbb{B}}_{ij}|^{2}}(x)=\sum_{r=0}^{S_{\mathbb{B} }}\binom{S_{\mathbb{B} }}{r}(-1)^{r}e^{-r\xi _{\mathbb{B}}x}.
\end{aligned}
\end{equation}

\subsection{Noise Uncertainty Model}\label{sec_noise_uncertainty}

In practice, the background noise is uncertain due to its complex composition (including thermal noise, quantization noise, imperfect filters, ambient wireless signals, etc.), dynamic environment and calibration error. Therefore, it is practical to use noise uncertainty to realize covert communication in mobile communication scenarios. In this work, we adopt the typical bounded noise uncertainty model \cite{shellhammer2006performance}, where the exact noise power of node $i$ $\sigma _{i}^{2}$ lies in a finite range around the nominal noise power $\hat{\sigma} _{n}^{2}$, then, the log-uniform distribution of $\sigma _{i}^{2}$ for the bounded uncertainty model is given by
\begin{align} \label{Eq:noise_un_pdf}
f_{\sigma _{i}^{2}}(x)=\begin{cases}
  \frac{1}{2\ln{(\rho)}x}, & \frac{1}{\rho}\hat{\sigma} _{n}^{2}\leq \sigma _{i}^{2}\leq \rho\hat{\sigma}_{n}^{2},\\
 0, & otherwise,
\end{cases}
\end{align}
where $\rho$ is the parameter that quantifies the level of the uncertainty.

\subsection{Definitions}\label{param_definitions}
Some basic definitions involved in this work are as follows.

\textbf{Detection Error Probability (DEP) $\bm{P_{ew}}$:} Let $H_{0}$ (the null hypothesis and $H_{1}$ (the alternative hypothesis) denote that Alice does not transmit and transmits covert information, respectively, let $D_{0}$ and $D_{1}$ respectively indicate the judgment of Willie that UAV is transmitting or not transmitting covert information. The DEP $P_{ew}$ is defined as the probability that Willie makes a wrong decision on whether or not UAV is transmitting covert messages, which is expressed as
\begin{equation}
P_{ew}=P_{FA}+P_{MD}\label{Eq:Pew},
\end{equation}
where $P_{FA}=\mathbb{P}(D_{1}|H_{0})$ is the false alarm probability that Willie is in favor of $H_{1}$ while $H_{0}$ is true, $P_{MD}=\mathbb{P}(D_{0}|H_{1})$ is the missed detection probability that Willie is in favor of $H_{0}$ while $H_{1}$ is true.

\textbf{Outage Probability $\bm{P_{out}}$:} The outage probability $\bm{P_{out}}$ is defined as the probability that the channel capacity $C_{ab}=W\log_2(1+\gamma_{ab})$ cannot support a given target transmission rate $R_{b}$, i.e., $P_{out}=\mathbb{P}(C_{ab}<R_{b})$ which can also be denoted as $P_{out}=\mathbb{P}(\gamma_{ab}<\gamma_{th})$ where $W$ is the bandwidth of $a\to b$ link, $\gamma_{ab}$ is the SNR, and $\gamma_{th}$ is the threshold of SNR which can be calculated as $\gamma_{th}=2^{\frac{R_{b}}{W}}-1$.

\textbf{Effective Covert Rate (ECR) $\bm{R_{ab}}$:} The effective covert rate $R_{ab}$ is defined the average \emph{successfully} transmitted amount of information, which can be denoted as $R_{ab}= R_{b} \times (1-P_{out})$, where $R_{b}$ is the target covert transmission rate and $P_{out}$ is the outage probability.

\textbf{Covert Shannon (Ergodic) Capacity (CSC) $\bm{C_{ab}}$:} The covert Shannon capacity is defined as the maximum value of the time average rate of messages delivered from transmitter to the destination, which is expressed as $C_{ab}=\mathbb{E}[W\log_2(1+\gamma_{ab})]$, where $W$ is the bandwidth and $\gamma_{ab}$ is the SNR.


\section{Transmission Mode and Detection Strategy} \label{sec3}

In this section, we first introduce two transmission modes of UAV (OM transmission mode and DM transmission mode), and then present the detection strategy at Willie.

\subsection{OM Transmission Mode}
When UAV adopts the OM transmission mode, the signals $y_{b}(k)$ ($k \in\{ 1, 2, \cdots, N\}$) received by Bob at $k$-th channel use is given by
\begin{equation} \label{2eq2_bob}
\begin{aligned}
y_{b}^{o}(k) =\sqrt{P_{a}}H_{ab}^{o,\mathbb{B}}x_{c}(k)+n_{b}(k),
\end{aligned}
\end{equation}
and the signals $y_{w}^{o}(k)$ received by Willie is given by
\begin{equation} \label{2eq2}
\begin{aligned}
y_{w}^{o}(k) =\sqrt{P_{a}}H_{aw}^{o,\mathbb{B}}x_{c}(k)+n_{w}(k).
\end{aligned}
\end{equation}
where $x_{c}$ is the desired signal, which following a zero-mean Gaussian distribution, i.e., $\mathbb{E}[|x_{c}(k)|^{2}]=1$. Besides, $n_{b}(k)$ and $n_{w}(k)$ are the received noise power at Bob and Willie, respectively, and the PDF is given by (\ref{Eq:noise_un_pdf}).

\subsection{DM Transmission Mode}
When UAV adopts the DM transmission mode, the signals $y_{b}^{d}(k)$ received by Bob at the $k$-th ($k \in \{1, 2, \cdots, N\}$) channel use is given by
\begin{equation} \label{eq2_bob}
\begin{aligned}
y_{b}^{d}(k) =\sqrt{P_{a}G_{ab}}H^{d,\mathbb{B}}_{ab}x_{c}(k)+n_{b}(k),
\end{aligned}
\end{equation}
and the signals $y_{w}^{d}(k)$ received by Willie is given by
\begin{equation} \label{eq2}
\begin{aligned}
y_{w}^{d}(k) =\sqrt{P_{a}G_{aw}}H^{d,\mathbb{B}}_{aw}x_{c}(k)+n_{w}(k),
\end{aligned}
\end{equation}
The desired signal $x_{c}$ and noise power obey the same distribution as the OM transmission mode.

\subsection{Detection Strategy at Willie}
Based on the observations over all time slots, Willie attempts to determine whether Alice conducts transmission or not. According to \cite{shahzad2017covert}, the optimal decision rule at Willie to minimize the detection error probability is determined as
\begin{equation} \label{neq2}
\begin{aligned}
T_{w}=\lim_{N\rightarrow \infty}\frac{P_{tol}}{N}\mathop{\gtrless}\limits_{H_0}^{H_1} \tau,
\end{aligned}
\end{equation}
\noindent where $P_{tol}=\textstyle\sum_{k=1}^{N}|y_{w}(k)|^{2}$ is the total power received by Willie in a given block, $\tau$ is the decision threshold, and $N$ is the total number of channel uses in a block. According to (\ref{2eq2}), when Alice adopts the OM transmission mode, the received signal $y_{w}^{o}(k)$ at Willie is given by
	\begin{align} \label{Eq:Tw_TA}
	T_{w}^{o}\!=\!\begin{cases}
	 \sigma _{w}^{2},\!\!\! &\text{$H_0$} \\
	 P_{\!a}L_{aw}^{o,\mathbb{B}}|h^{o,\mathbb{B}}_{aw}|^{2}\!+\sigma _{w}^{2}.\!\!\!&\text{$H_1$}
	\end{cases}
	\end{align}
Similarly, according to (\ref{eq2}), when Alice adopts the DM transmission mode, the received signal $y_{w}^{d}(n)$ received at Willie is given by
	\begin{align} \label{Eq:Tw}
	T_{w}^{d}\!=\!\!\begin{cases}
	  \sigma _{w}^{2},\!\! &\text{$H_0$} \\
	  P_{\!a}G_{\!aw}L_{aw}^{d,\mathbb{B}}|h^{d,\mathbb{B}}_{aw}|^{2}\!+\sigma_{w}^{2}. \!\!&\text{$H_1$}
	\end{cases}
	\end{align}

\section{Covert Performance Analysis under the OM Transmission Mode}\label{TA_Performance}

In this section, we derive the optimal detection threshold and the corresponding minimum DEP at Willie, as well as the expected minimum DEP from Alice’s perspective, lastly, the ECR and CSC under the OM transmission mode.

\subsection{Detection Error Probability}
To determine the DEP at Willie, we first analyze the false alarm probability and the missed detection probability. Based on (\ref{neq2}) and (\ref{Eq:Tw_TA}), the false alarm probability $P_{FA}^{o}$ is given by
\begin{align} \label{2eq5_NJ}
P_{FA}^{o} =\mathbb{P}(T_{w}^{o}>\tau |H_{0})=\mathbb{P}\left(\sigma _{w}^{2}\frac{\chi_{2N}^{2}}{N}>\tau\right),
\end{align}
where $\chi_{2N}^{2}$ denotes a random value that follows chi-squared distribution with $2N$ degrees of freedom. From the strong law of large numbers, we know that $\chi_{2N}^{2}/N$ converges to $1$. Based on Lebesgue's Dominated Convergence Theorem \cite{browder2012mathematical}, $P_{FA}^{o}$ can be derived as
\begin{align} \label{eq5_NJ}
P_{FA}^{o} & =\mathbb{P}(\sigma _{w}^{2}>\tau ) \notag \\
& =\begin{cases}
   1, & \tau < \frac{\hat{\sigma}_{n}^{2}}{\rho}, \\
  1-\frac{\ln(\rho\tau)-\ln(\hat{\sigma}_{n}^{2})}{2\ln(\rho)}, &  \frac{\hat{\sigma}_{n}^{2}}{\rho}\leq \tau \leq \rho\hat{\sigma}_{n}^{2}, \\
   0, & \tau > \rho\hat{\sigma}_{n}^{2}.
     \end{cases}
\end{align}
Similarly, the missed detection probability can be derived as
\begin{align} \label{eq6_NJ}
	P_{MD}^{o}\!&=\mathbb{P}(T_{w}^{o}>\tau |H_{1}) =\mathbb{P}(k_{a}^{o}+\sigma _{w}^{2}<\tau ) \notag\\
     &=\begin{cases}
				1, & \tau \!>\!k_{a}^{o}\!+\!\rho\hat{\sigma}_{n}^{2},\\
				\frac{\ln(\rho(\tau -k_{a}^{o}))-\ln(\hat{\sigma}_{n}^{2})}{2\ln(\rho)},& k_{a}^{o}\!+\!\frac{\hat{\sigma}_{n}^{2}}{\rho}\leq \tau \!\leq \! k_{a}^{o}\!+\!\rho\hat{\sigma}_{n}^{2},\\
				0 , & \tau\!<\! k_{a}^{o}\!+\!\frac{\hat{\sigma}_{n}^{2} }{\rho},
				\end{cases}
\end{align}
where $k_{a}^{o}=P_{a}L_{aw}^{o,\mathbb{B}}|{h}^{o,\mathbb{B}}_{aw}|^{2}$. According to (\ref{eq5_NJ}) and (\ref{eq6_NJ}), we can analyze the optimal detection threshold and the minimization DEP for Willie under OM transmission mode as follows.

\begin{lemma} \label{2th1_NJ}
When Alice adopts the OM transmission mode, the optimal threshold $\tau^{*}$ for Willie's detector is in the interval
	\begin{equation} \label{2eq3_NJ}
	\begin{aligned}
	\tau^{*}\in\begin{cases}
	 [\rho\hat{\sigma}_{n}^{2},k_{a}^{o}+\frac{\hat{\sigma}_{n}^{2}}{\rho}],& \rho\hat{\sigma}_{n}^{2}< k_{a}^{o}+\frac{\hat{\sigma}_{n}^{2}}{\rho},\\
	 k_{a}^{o}+\frac{\hat{\sigma}_{n}^{2}}{\rho} , & \rho\hat{\sigma}_{n}^{2}\geq k_{a}^{o}+\frac{\hat{\sigma}_{n}^{2}}{\rho},
	\end{cases}
	\end{aligned}
	\end{equation}
and the corresponding minimum DEP $P_{ew}^{*,o}$ is given as
	\begin{equation} \label{2eq4_NJ}
	\begin{aligned}
	P_{ew}^{*,o}=\begin{cases}
	0 , & \rho\hat{\sigma}_{n}^{2}< k_{a}^{o}+\frac{\hat{\sigma}_{n}^{2}}{\rho}, \\
	1-\frac{\ln(\rho k_{a}^{o}+\hat{\sigma}_{n}^{2})-\ln( \hat{\sigma}_{n}^{2})}{2\ln(\rho)} , & \rho\hat{\sigma}_{n}^{2} \geq k_{a}^{o}+\frac{\hat{\sigma}_{n}^{2}}{\rho},
	\end{cases}
	\end{aligned}
	\end{equation}
where $k_{a}^{o}=P_{a}L_{aw}^{o,\mathbb{B}}|{h}^{o,\mathbb{B}}_{aw}|^{2}$, $\rho$ and $\hat{\sigma}_{n}^{2}$ are the parameter that quantifies the size of the uncertainty and nominal noise power, respectively, defined in \ref{sec_noise_uncertainty}.
\end{lemma}

\begin{proof}
To find the optimal threshold, we consider the optimization problem ${\min}_{\tau} \ P_{ew}^{o}=P_{FA}^{o}+P_{MD}^{o}$. According to (\ref{eq5_NJ}) and (\ref{eq6_NJ}), the following two cases are considered.

\emph{Case I}: When $\rho\hat{\sigma}_{n}^{2}< k_{a}^{o}+\frac{\hat{\sigma}_{n}^{2}}{\rho}$, for any $\tau \in [\rho\hat{\sigma}_{n}^{2},k_{a}^{o}+\frac{\hat{\sigma}_{n}^{2}}{\rho}]$, we have $P_{ew}^{o}=P_{FA}^{o}+P_{MD}^{o}=0$.

\emph{Case II}: When $\rho\hat{\sigma}_{n}^{2}\geq k_{a}^{o}+\frac{\hat{\sigma}_{n}^{2}}{\rho}$, the DEP $P_{ew}^{o}=P_{FA}^{o}+P_{MD}^{o}$ can be written as
\setlength\arraycolsep{1pt}
\begin{footnotesize}
\begin{subnumcases}
{P_{ew}^{o}=\!}
1, &  $\tau<\frac{\hat{\sigma}_{n}^{2} }{\rho}$, \qquad\label{eq7a_NJ} \\
1\!- \!\frac{\ln(\rho\tau)\!-\!\ln(\hat{\sigma}_{n}^{2})}{2\ln(\rho)}, & $\frac{\sigma_{n}^ {2}}{\rho}\! \leq \!\tau < k_{a}^{o}\!+\!\frac{\hat{\sigma}_{n}^{2}}{\rho}$, \qquad\label{eq7b_NJ} \\
1\!-\!\frac{\ln(\tau)\!-\!\ln(\tau-k_{a}^{o})}{2\ln(\rho)},& $k_{a}^{o}\!+\!\frac{\hat{\sigma}_{n}^{2}}{\rho}\! \leq \!\tau \!\leq\!\rho\hat{\sigma}_{n}^{2}$, \qquad\label{eq7c_NJ} \\
\frac{\ln(\rho(\tau\! -\! k_{a}^{o}))\!-\!\ln(\hat{\sigma}_{n}^{2})}{2\ln(\rho)},& $\rho\hat{\sigma}_{n}^{2}\!< \!\tau \!\leq \! k_{a}^{o}\!+\!\rho\hat{\sigma}_{n}^{2}$, \qquad\label{eq7d_NJ} \\
1, & $\tau >k_{a}^{o}\!+\!\rho\hat{\sigma}_{n}^{2}$ . \qquad\label{eq7e_NJ}
\end{subnumcases}
\end{footnotesize}\normalsize
then we can see from (\ref{eq7b_NJ}), when $\frac{\sigma_{n}^ {2}}{\rho} \leq \tau < k_{a}^{o}+\frac{\hat{\sigma}_{n}^{2}}{\rho}$, $P_{ew}^{o}$ is a decreasing function of $\tau$. Thus, when $\tau = k_{a}^{o}+\frac{\hat{\sigma}_{n}^{2}}{\rho}$, $P_{ew}^{o}$ gets the minimum value which equals to (\ref{eq7c_NJ}). From (\ref{eq7d_NJ}), when $\rho\hat{\sigma}_{n}^{2}< \tau \leq k_{a}^{o}+\rho\hat{\sigma}_{n}^{2}$, $P_{ew}^{o}$ is a increasing function of $\tau$. Thus, when $\tau = \rho\hat{\sigma}_{n}^{2}$, $P_{ew}^{o}$ gets the minimum value which equals to (\ref{eq7c_NJ}). Overall, the DEP gets the minimum value $P_{ew}^{*,o}$ for Willie if and only if $\tau\in[k_{a}^{o}\!+\!\frac{\hat{\sigma}_{n}^{2}}{\rho},\rho\hat{\sigma}_{n}^{2}]$. After deriving the first-order derivative of $P_{ew}^{o}$ with respect to $\tau$, we have
\begin{align}\label{Eq:pew_derivative}
 \frac{\partial P_{ew}^{o}}{\partial \tau} =\frac{k_{a}^{o}}{2\tau(\tau-k_{a}^{o})\ln{(\rho)}} >0.
\end{align}
Obviously, $P_{ew}^{o}$ is monotonically increasing in the interval. Thus, the optimal threshold $\tau^{*}$ for Willie is $\tau= k_{a}^{o}+\frac{\hat{\sigma}_{n}^{2}}{\rho}$ and the corresponding minimum DEP $P_{ew}^{*,o}$ is
\begin{align}\label{Eq:minimum_pew}
P_{ew}^{*,o}=1-\frac{\ln(\rho k_{a}^{o}+\hat{\sigma}_{n}^{2})-\ln( \hat{\sigma}_{n}^{2})}{2\ln(\rho)}.
\end{align}
Based on Case I and Case II, we have (\ref{2eq3_NJ}) and (\ref{2eq4_NJ}).
\end{proof}

Since only the statistical channel static information of Alice $\to$ Willie, we use the expected measure of $P_{ew}^{*,o}$ to evaluate the covertness. From (\ref{2eq4_NJ}), we note that $\mathbb{E}[P_{ew}^{*,o}]$ is related to the numerical integration of $|h^{o,\mathbb{B}}_{aw}|^{2}$. To this end, we first present the following lemma.
\begin{lemma}\label{lemma:calu2}
If $h^{o,\mathbb{B}}_{aw}$ $(\mathbb{B}\in \{L,N\})$ is the channel fading coefficient of the microwave channel, we have
\begin{scriptsize}
\begin{align}\label{Eq:calu2}
\int_{0}^{a}\!{x}f_{|h^{ o,\mathbb{B}}_{ aw}|^{2}}({x})d{x}\!&=\frac{\gamma\left ( \frac{2}{\nu(\sqrt{2k_{\mathbb{B}}})}, \left[2a(k_{\mathbb{B}}\!+\!1)\right]^\frac{\nu(\sqrt{2k_{\mathbb{B}}})}{2}e^{\mu(\sqrt{2k_{\mathbb{B}}})} \right )}{(k_{\mathbb{B}}+\!1)\nu(\sqrt{2k_{\mathbb{B}}})e^{\frac{2\mu(\sqrt{2k_{\mathbb{B}}})}{\nu(\sqrt{2k_{\mathbb{B}}})}}} \notag \\
& \qquad-a\exp{\left(\!-\left[2a(k_{\mathbb{B}}\!+\!1)\right]^\frac{\nu(\sqrt{2k_{\mathbb{B}}})}{2}e^{\mu(\sqrt{2k_{\mathbb{B}}})}\right)},
\end{align}
\end{scriptsize}\normalsize
where $a\ge0$, $f_{|h^{ o,\mathbb{B}}_{ aw}|^{2}}$ is the PDF of ${|h^{ o,\mathbb{B}}_{ aw}|^{2}}$, $\gamma(\cdot,\cdot)$ is the lower incomplete gamma function, $\mu(x)$ and $\nu(x)$ are the polynomial expressions of $x$ as in \cite{mishra20162,bocus2013approximation}, which are respectively given by
\begin{scriptsize}
\begin{align}\label{Eq:s2_rician_q_u_NJ}
 \mu(x) \triangleq\! \begin{cases}
   -\ln2, & x=0 , \\
   -3.0888\! \times \!\!10^{-10}x^{6}\!+\!1.8362 \!\times \!\!10^{-7}\!x^{5} \\
   -3.7185\! \times \!\!10^{-5}x^{4} \!+\!3.4103 \!\times \!\!10^{-3}x^{3} \\
	 -0.1624\! \times \!\!x^{2} \!-\!1.4318x\!+\!0.7409, & 10 \leq x \leq 8000 ,
     \end{cases}
\end{align}
\end{scriptsize}\normalsize
\vspace*{-0.4cm}
\begin{scriptsize}
\begin{align}\label{Eq:s2_rician_q_v_NJ}
 \nu (x) \triangleq \begin{cases}
   2, & x=0 , \\
    5.1546 \!\times \!10^{-11}x^{6}\!-\!3.1961\!\times\!10^{-8}x^{5} \\
	 +6.3859 \!\times \!10^{-6}x^{4}\!-\!5.4159\!\times \!\!10^{\!-4}x^{3} \\
	 +1.9833 \!\times \!10^{-2}x^{2}\!+\!0.9044x\!+\!0.9439 , & 10 \leq x \leq 8000 .
     \end{cases}
\end{align}
\end{scriptsize}
\end{lemma}

\begin{proof}
The detailed proof is presented in Appendix~\ref{nnap1}.
\end{proof}

Considering the uncertainty of the LoS/NLoS channels, the probability $P^{\mathbb{B}}_{aw}$ should be considered in the calculation of $\mathbb{E}[P_{ew}^{*,o}]$. According to Lemma~\ref{lemma:calu2}, we can derive $\mathbb{E}[P_{ew}^{*,o}]$ as the following theorem.

\begin{theorem} \label{2th2_NJ}
When Alice adopts the OM transmission mode, the expected value $\mathbb{E}[P_{ew}^{*,o}]$ from Alice's perspective is given as (\ref{2eqn1_NJ}),
\begin{figure*}[t]
\begin{small}
\begin{equation} \label{2eqn1_NJ}
\mathbb{E}[P_{ew}^{*,o}]=\smashoperator{\sum_{\mathbb{B}\in\{L,N\}}}P_{aw}^{\mathbb{B}} (1-\Theta_\mathbb{B}^{o}) \left\{1 - \frac{1}{2\ln(\rho)}\left\{\ln\!\left[\rho P_{a}L_{aw}^{o,\mathbb{B}} \left( \frac{\!\gamma \!\left ( \frac{2}{\nu(\sqrt{2k_{\mathbb{B}}})}, [2\varrho^{o}(k_{\mathbb{B}}\!+\!1)]^\frac{\nu(\sqrt{2k_{\mathbb{B}}})}{2}e^{\mu(\sqrt{2k_{\mathbb{B}}})} \right )}{(k_{\mathbb{B}}\!+\!1)\nu(\sqrt{2k_{\mathbb{B}}})e^{\frac{2\mu(\sqrt{2k_{\mathbb{B}}})}{\nu(\sqrt{2k_{\mathbb{B}}})}}} -\varrho^{o} \Theta_\mathbb{B}^{o}\!\right ) \!+\!\hat{\sigma}_{n}^{2} \right]\!-\ln( \hat{\sigma}_{n}^{2}) \!\right\} \right\}.
\end{equation}
\end{small}
\rule[1cm]{\textwidth}{0.04em}
\vspace*{-0.6cm}
\end{figure*}
where $\varrho^{o}=\frac{(\rho^{2}-1)\hat{\sigma}_{n}^{2}}{\rho P_{a}L_{aw}^{o,\mathbb{B}}}$ and $\Theta_\mathbb{B}^o$ is defined as
\begin{equation} \label{eqn22}
\Theta_\mathbb{B}^{o}= \exp{\left(-\left[2\varrho^{o}(k_{\mathbb{B}}+1)\right]^\frac{\nu(\sqrt{2k_{\mathbb{B}}})}{2}e^{\mu(\sqrt{2k_{\mathbb{B}}})}\right)}.
\end{equation}
\end{theorem}

\begin{proof} For convenience, let's write $k_{1}^{o}=\rho\hat{\sigma}_{n}^{2}$ and $k_2^{o}=k_{a}^{o}+\frac{\hat{\sigma}_{n}^{2}}{\rho}$. According to (\ref{2eq4_NJ}) in Lemma~\ref{2th1_NJ}, we have
\begin{align} \label{2eq8_NJ}
\mathbb{E}&[P_{ew}^{*,o}] \!=\!\mathbb{E}_{k_{1}^{o}< k_{2}^{o}}[P_{ew}^{*,o}]\mathbb{P}(k_{1}^{o}<k_{2}^{o})\!+\!\mathbb{E}_{k_{1}^{o}\geq k_{2}^{o}}[P_{ew}^{*,o}]\mathbb{P}(k_{1}^{o}\!\geq\! k_{2}^{o}) \notag \\
&\!=\!\!\mathbb{P}(k_{1}^{o}\!\geq\! k_{2}^{o}\!)\!\!\left(\!\!1\!\!-\!\!\frac{\ln(\rho P_{a}L_{aw}^{o,\mathbb{B}}\mathbb{E}_{k_{1}^{o}\geq k_{2}^{o}}\!\!\left[|{h}^{o,\mathbb{B}}_{aw}|^{2}\right]\!\!+\!\hat{\sigma}_{n}^{2})\!-\!\ln(\hat{\sigma}_{n}^{2})}{2\ln(\rho)}\!\right)\!,
\end{align}
where $\mathbb{P}(k_{1}^{o}\geq k_{2}^{o})$ can be derived as
\begin{align} \label{2eq9_NJ}
\mathbb{P}(k_{1}^{o}\geq k_{2}^{o})&=\mathbb{P}\left(\rho\hat{\sigma}_{n}^{2} \geq P_{a}L_{aw}^{o,\mathbb{B}}|{h}^{o,\mathbb{B}}_{ aw}|^{2}+\frac{\hat{\sigma}_{n}^{2}}{\rho}\right)\notag\\
&=\mathbb{P}\left(|{h}^{o,\mathbb{B}}_{aw}|^{2}\leq \varrho^{o} \right) \notag\\
&=1-Q_{1}\left( \sqrt{2k_{\mathbb{B}}},\sqrt{2(k_{\mathbb{B}}+1)\varrho^{o}} \right).
\end{align}

To provide some deep analytical insights on the average value of $P_{ew}^{*,o}$, a tight exponential-type approximation for the standard Marcum-Q function $Q_{1}(\cdot,\cdot)$ is adopted, which is expressed by
\begin{align} \label{Eq:s2_rician_q_approx}
Q_{1}(x,y)\approx \exp{(-e^{\mu(x)}y^{\nu(x)})},
\end{align}
where $\mu(x)$ and $\nu(y)$ are conditionally defined for $0\leq x \ll 1$ and $10\leq x \leq 8000$, which is good range for Rice factors $k_{\mathbb{B}}$ in the concerned system. In this work, we define $\mu(x)$ and $\nu(x)$ as (\ref{Eq:s2_rician_q_u_NJ}) and (\ref{Eq:s2_rician_q_v_NJ}), respectively. The quality and reliability of approximation has been testified where the root mean square error (RMSE) of approximation is less than $0.005$ \cite{bocus2013approximation}. Thus, we can obtain an approximation of $Q_{1}(\cdot,\cdot)$ in (\ref{2eq9_NJ}) as
\begin{align}\label{2eq81_NJ}
\mathbb{P}&(k_{1}^{o}\!\geq\! k_{2}^{o})=1-\exp{\left(-e^{\mu(\sqrt{2k_{\mathbb{B}}})}\left[2\varrho^{o} (k_{\mathbb{B}}+1)\right]^\frac{\nu(\sqrt{2k_{\mathbb{B}}})}{2}\right)}.
\end{align}

For the expectation term in (\ref{2eq8_NJ}), we have
\begin{align}\label{2eq11_NJ}
\mathbb{E}&_{k_{1}^{o}\geq k_{2}^{o}} \left[ |{h}^{o,\mathbb{B}}_{aw}|^{2} \right]=\mathbb{E}\left [ |{h}^{o,\mathbb{B}}_{aw}|^{2}\Bigg||{h}^{o,\mathbb{B}}_{aw}|^{2}\leq \varrho^{o} \right ] \notag \\
&=\int_{0}^{\varrho }{x}f_{|{h}^{o,\mathbb{B}}_{aw}|^{2}}({x})d{x} \notag \\
&\overset{(a)}{=}\frac{\gamma\left( \frac{2}{\nu(\sqrt{2k_{\mathbb{B}}})}, \left[2\varrho^{o}(k_{\mathbb{B}}\!+\!1)\right]^\frac{\nu(\sqrt{2k_{\mathbb{B}}})}{2}e^{\mu(\sqrt{2k_{\mathbb{B}}})} \right)}{\nu(\sqrt{2k_{\mathbb{B}}})(k_{\mathbb{B}}+1)e^{\frac{2\mu(\sqrt{2k_{\mathbb{B}}})}{\nu(\sqrt{2k_{\mathbb{B}}})}}} \notag \\
& \qquad-\varrho^{o}\exp{\left(-[2\varrho^{o}(k_{\mathbb{B}}+1)]^\frac{\nu(\sqrt{2k_{\mathbb{B}}})}{2}e^{\mu(\sqrt{2k_{\mathbb{B}}})}\right)},
\end{align}
where step (a) is due to Lemma~\ref{lemma:calu2}. Finally, submitting (\ref{2eq81_NJ}) and (\ref{2eq11_NJ}) into (\ref{2eq8_NJ}), we can get $\mathbb{E}[P_{ew}^{*,o}]$ as (\ref{2eqn1_NJ}).
\end{proof}

\subsection{Effective Covert Rate and Covert Shannon Capacity}\label{oa_ecrcsc}

According to the definition mentioned in Section~\ref{param_definitions}, the ECR $R_{ab}^{o}$ can be given as the following theorem.

\begin{theorem}\label{th3_s2_out_NJ}
When UAV adopts the OM transmission mode in the concerned A2G system, the ECR $R_{ab}^{o}$ of the system is determined by
\begin{align}\label{Eq:s2_out_NJ}
R_{ab}^{o}\!=\!R_{b}\!\times\! \left\{\! 1\!-\smashoperator{\sum_{\mathbb{B}\in\{L,N\}}}P_{ab}^{\mathbb{B}}\!
 \left[ 1\! - \!\mathcal{F}_\mathrm{Ei}^{o}\!(\rho\hat{\sigma}_{n}^{2}) \!+ \!\mathcal{F}_\mathrm{Ei}^{o} \!\left( \! \frac{\hat{\sigma}_{n}^{2}}{\rho} \right)\! \right]\!\right\},
\end{align}
where $R_{b}$ is the target rate, $\mathcal{F}_{\mathrm{Ei}}^{o}(\cdot)$ is defined as following
\begin{small}
\begin{align}\label{Eq:s1_fp}
\mathcal{F}_{\mathrm{Ei}}^{o}(x) \!= \!\frac{1}{\nu(\sqrt{2k_{\mathbb{B}}})\ln{\rho}}\mathrm{Ei}\!\!\left[ \!-e^{\mu(\sqrt{2k_{\mathbb{B}}})} \! \!\left (\!\frac{2(k_{\mathbb{B}}\!+\!1)\gamma_{th}x}{P_aL_{ab}^{o,\mathbb{B}}}\!\right)^{\!\!\!\frac{\nu(\sqrt{2k_{\mathbb{B}}})}{2}} \!\right],
\end{align}
\end{small}\normalsize
$\mu(x)$ and $\nu(x)$ are defined as (\ref{Eq:s2_rician_q_u_NJ}) and (\ref{Eq:s2_rician_q_v_NJ}) in Theorem \ref{2th2_NJ}, respectively, and $\mathrm{Ei}(\cdot)$ is the exponential integral function defined as \cite[Eq. (8.211.1)]{gradshteyn2014table}.
\end{theorem}

\begin{proof}
To derive the ECR, we first need to determined the outage probability at Bob $P^{o}_{out}$. Considering the uncertainty of the LoS/NLoS channel, the outage probability at Bob $P^{o}_{out}$ under the OM transmission mode can be determined by
\begin{equation}\label{Eq:s2_proof_out11_NJ}
P_{out}^{o}=\,\,\smashoperator{\sum_{\mathbb{B}\in\{L,N\}}}P_{ab}^{\mathbb{B}}\times\mathbb{P}\left(\gamma_{ab}^{o}<\gamma_{th}\right).
\end{equation}
 According to (\ref{2eq2_bob}), the SNR at Bob is determined by $\gamma _{ab}^{o}={P_a L_{ab}^{o,\mathbb{B}}|h^{o,\mathbb{B}}_{ab}|^{2}}/{\sigma _{b}^{2}}$.
Thus, the outage probability at Bob $P^{o}_{out}$ is given by
\begin{footnotesize}
\begin{align}\label{Eq:s2_proof_out_TANJ}
&\mathbb{P}\left(\gamma_{ab}^{o}<\gamma_{th}\right)
=\mathbb{P}\left(|{h}^{o,\mathbb{B}}_{ab}|^{2}<\frac{\sigma _{b}^{2}\gamma_{th}}{P_aL_{ab}^{o,\mathbb{B}}}\right) \notag \\
&\overset{(b)}{=} 1-\mathbb{E}_{\sigma _{b}^{2}} \left [\exp{\left(-e^{\mu(\sqrt{2k_{\mathbb{B}}})} \left ( \frac{2(k_{\mathbb{B}}+1)\sigma _{b}^{2}\gamma_{th}}{P_aL_{ab}^{T,\mathbb{B}}} \right)^{\!\!\frac{\nu(\sqrt{2k_{\mathbb{B}}})}{2}}\right)} \right ] \notag \\
&=1- \int_{\frac{\hat{\sigma}_{n}^{2}}{\rho}}^{\rho\hat{\sigma}_{n}^{2}}\exp{\left(-e^{\mu(\sqrt{2k_{\mathbb{B}}})}\! \left( \!\frac{2(k_{\mathbb{B}}+1)\gamma_{th}x}{P_aL_{ab}^{o,\mathbb{B}}}\! \right)^{\!\!\frac{\nu(\sqrt{2k_{\mathbb{B}}})}{2}}\right)} \frac{1}{2x\ln{\rho} }dx \notag \\
&\overset{(c)}{=} 1\!-\!\frac{1}{\nu(\sqrt{2k_{\mathbb{B}}})\ln{\rho}}\!\left[\mathrm{Ei}\!\left(\!\!-e^{\mu(\sqrt{2k_{\mathbb{B}}})} \!\left (\! \frac{2(k_{\mathbb{B}}\!+\!1)\gamma_{th}x}{P_aL_{ab}^{o,\mathbb{B}}} \!\right)^{\!\!\!\frac{\nu(\sqrt{2k_{\mathbb{B}}})}{2}}\! \right)\!\right]\!\Bigg|_{\frac{\hat{\sigma}_{n}^{2}}{\rho}}^{\rho\hat{\sigma}_{n}^{2}} \notag \\
&=1 -\frac{1}{\nu(\sqrt{2k_{\mathbb{B}}})\ln{\rho}} \mathrm{Ei}\left(-e^{\mu(\sqrt{2k_{\mathbb{B}}})} \left ( \frac{2(k_{\mathbb{B}}+1)\gamma_{th}\rho\hat{\sigma}_{n}^{2}}{P_aL_{ab}^{o,\mathbb{B}}} \right)^{\!\!\frac{\nu(\sqrt{2k_{\mathbb{B}}})}{2}} \right)  \notag\\
& \qquad+\frac{1}{\nu(\sqrt{2k_{\mathbb{B}}})\ln{\rho}} \mathrm{Ei}\left(-e^{\mu(\sqrt{2k_{\mathbb{B}}})} \left ( \frac{2(k_{\mathbb{B}}+1)\gamma_{th}\hat{\sigma}_{n}^{2}}{P_aL_{ab}^{o,\mathbb{B}}\rho} \right)^{\!\!\frac{\nu(\sqrt{2k_{\mathbb{B}}})}{2}} \right),
\end{align}
\end{footnotesize}\normalsize
where step (b) is similar to (\ref{2eq81_NJ}), and step (c) is due to $\int e^{ax^{n}}x^{-1}dy=\mathrm{Ei}(ax^{n})/n$ as \cite[Eq. (2.325.7)]{gradshteyn2014table}. Then, submitting (\ref{Eq:s2_proof_out_TANJ}) into (\ref{Eq:s2_proof_out11_NJ}), we can obtain (\ref{Eq:s2_out_NJ}).
\end{proof}

Next, we derive the covert Shannon capacity $C_{ab}^{o}$ of the concerned system under the OM transmission mode.
\begin{theorem} \label{th_s2_covert_rate_TANJ}
When UAV adopts the OM transmission mode in the concerned A2G system, the CSC $C_{ab}^{o}$ of the system is determined by
\begin{small}
\begin{align}\label{s2_rab_TANJ}
C_{ab}^{o}=\,\,\smashoperator{\sum_{\mathbb{B}\in\{L,N\}}}\,\,\,
\frac{W^{o}P_{ab}^{\mathbb{B}}}{2\ln2\times\ln\rho}\left [ \mathcal{F}_{\mathbf{Li}}^{o}\left(\frac{1}{\rho\hat{\sigma}_{n}^{2}}\right)-\mathcal{F}_{\mathbf{Li}}^{o}\left (\frac{\rho}{\hat{\sigma}_{n}^{2}}\right) \right],
\end{align}
\end{small}\normalsize
where $W^{o}$ is the bandwidth of microwave, $\mathcal{F}_{\mathbf{Li}}^{o}(\cdot)$ is given as
\begin{align}\label{Eq:s1_fr}
\mathcal{F}_{\mathbf{Li}}^{o}(x)= \int_{0}^{\infty}\mathbf{Li}_{2} (-P_aL_{ab}^{o,\mathbb{B}}xy)f_{|h^{o,\mathbb{B}}_{ab}|^{2}}(y)dy,
\end{align}
$\mathbf{Li}_{s}(\cdot)$ is the polylogarithm of order s, and $f_{|h^{o,\mathbb{B}}_{ab}|^{2}}(y)$ is refer to (\ref{Eq:rician_pdf}).
\end{theorem}

\begin{proof}
Considering the possibility of the LoS/NLoS channel, the CSC $C_{ab}^{o}$ of the A2G system is determined by
\begin{equation}\label{s2_rab_proof1_NJ}
C_{ab}^{o}=\,\,\smashoperator{\sum_{\mathbb{B}\in\{L,N\}}}P_{ab}^{\mathbb{B}}\times\,C_{ab}^{o}.
\end{equation}
According to the SNR at Bob, $C_{ab}^{o}$ can be derived as
\begin{footnotesize}
\begin{align}\label{s2_rab_proof_NJ}
C_{ab}^{o}\!&= W^{o}\mathbb{E}_{|h^{o,\mathbb{B}}_{ab}|^{2},\sigma _{b}^{2}}\left [ \log_{2}{\left(1+\frac{P_a L_{ab}^{o,\mathbb{B}}|h^{o,\mathbb{B}}_{ab}|^{2}}{\sigma _{b}^{2}}\right)} \right ] \notag \\
&=W^{o}\mathbb{E}_{|h^{o,\mathbb{B}}_{ab}|^{2}}\left [\frac{1}{2\ln\rho}\int_{\frac{\hat{\sigma}_{n}^{2}}{\rho}}^{\rho\hat{\sigma}_{n}^{2}} \log_{2}\left(1+\frac{P_aL_{ab}^{o,\mathbb{B}}|h^{o,\mathbb{B}}_{ab}|^{2}}{x}\right)\frac{1}{x}dx \right ] \notag \\
&=W^{o}\mathbb{E}_{|h^{o,\mathbb{B}}_{ab}|^{2}}\left [\frac{1}{2\ln2\ln\rho}\int_{\frac{1}{\rho\hat{\sigma}_{n}^{2}}}^{\frac{\rho}{\hat{\sigma}_{n}^{2}}}\!\ln\left(1+P_aL_{ab}^{o,\mathbb{B}}|h^{o,\mathbb{B}}_{ab}|^{2} y\right)\frac{1}{y}dy\right ] \notag \\
&\overset{(c)}{=}W^{o}\mathbb{E}_{|h^{o,\mathbb{B}}_{ab}|^{2}}\!\left[\frac{P_aL_{ab}^{o,\mathbb{B}}|h^{o,\mathbb{B}}_{ab}|^{2}y}{2\ln2\ln\rho}[\Phi (-P_aL_{ab}^{o,\mathbb{B}}|h^{o,\mathbb{B}}_{ab}|^{2}y,2,1)]\Big|_{\frac{1}{\rho\hat{\sigma}_{n}^{2}}}^{\frac{\rho}{\hat{\sigma}_{n}^{2}}}\right ] \notag \\
&\overset{(d)}{=}W^{o}\mathbb{E}_{|h^{o,\mathbb{B}}_{ab}|^{2}}\left[\frac{-1}{2\ln2\ln\rho}\left[\mathbf{Li}_{2} (-P_aL_{ab}^{o,\mathbb{B}}|h^{o,\mathbb{B}}_{ab}|^{2}y)\right]\Big|_{\frac{1}{\rho\hat{\sigma}_{n}^{2}}}^{\frac{\rho}{\hat{\sigma}_{n}^{2}}}\right ] \notag \\
&=\frac{W^{o}}{2\ln2\ln\rho}\int_{0}^{\infty }\mathbf{Li}_{2} \left(\frac{-P_aL_{ab}^{o,\mathbb{B}}x}{\rho \hat{\sigma}_{n}^{2}}\right)f_{|h^{o,\mathbb{B}}_{ab}|^{2}}(x) dx \notag \\
&\,\,\,\, -\frac{W^{o}}{2\ln2\ln\rho}\int_{0}^{\infty }\mathbf{Li}_{2} \left(\frac{-\rho P_aL_{ab}^{o,\mathbb{B}}x}{\hat{\sigma}_{n}^{2}}\right)  f_{|h^{o,\mathbb{B}}_{ab}|^{2}}(x) dx,
\end{align}
\end{footnotesize}\normalsize
where step (c) is according to \cite[Eq. (2.728.2)]{gradshteyn2014table}, $\mathbf{\Phi}(\cdot,\cdot,\cdot)$ is Lerch function defined as \cite[Eq. (9.550)]{gradshteyn2014table}, and step (d) is due to $\mathbf{Li}_{s}(z)=z\mathbf{\Phi}(z,s,1)$. Submitting (\ref{s2_rab_proof_NJ}) into (\ref{s2_rab_proof1_NJ}), we can obtain (\ref{s2_rab_TANJ}).
\end{proof}


\section{Covert Performance Analysis under the DM Transmission Mode}\label{MA_Performance}

Similar to Section~\ref{TA_Performance}, this section investigates the optimal detection threshold and the corresponding minimum DEP at Willie, as well as the expected minimum DEP from Alice’s perspective, lastly, the ECR and CSC under the DM transmission mode.

\subsection{Detection Error Probability}
According to (\ref{neq2}) and (\ref{Eq:Tw}), when Alice does not transmit information, Willie only received the background noise. Thus, $P^{d}_{FA}$ is same as (\ref{eq5_NJ}). While Alice transmits information, by observing the equation deducing in (\ref{eq6_NJ}), we can obtain $P^{d}_{MD}$ by replacing the parameter $k_a^{o}$ in (\ref{eq6_NJ}) with $k_a^{d}=P_{a}G_{aw}L_{aw}^{d,\mathbb{B}}|{h}^{d,\mathbb{B}}_{aw}|^{2}$. Based on $P_{FA}^{d}$ and $P_{MD}^{d}$, we can further derive the optimal detection threshold and the minimum detection error probability at Willie as the following lemma.
\begin{lemma} \label{th1_NJ}
When Alice adopts the DM transmission mode, the optimal detection threshold $\tau^{*}$ for Willie's detector is in the interval
	\begin{equation} \label{eq3_NJ}
	\begin{aligned}
	\tau^{*}\in\begin{cases}
	 [\rho\hat{\sigma}_{n}^{2},k_{a}^{d}+\frac{\hat{\sigma}_{n}^{2}}{\rho}],& \rho\hat{\sigma}_{n}^{2}< k_{a}^{d}+\frac{\hat{\sigma}_{n}^{2}}{\rho},\\
	 k_{a}^{d}+\frac{\hat{\sigma}_{n}^{2}}{\rho},& \rho\hat{\sigma}_{n}^{2}\geq k_{a}^{d}+\frac{\hat{\sigma}_{n}^{2}}{\rho},
	\end{cases}
	\end{aligned}
	\end{equation}
and the corresponding minimum DEP $P_{ew}^{*,d}$ is given as
	\begin{equation} \label{eq4_NJ}
	\begin{aligned}
	P_{ew}^{*,d}=\begin{cases}
	0 , & \rho\hat{\sigma}_{n}^{2}< k_{a}^{d}+\frac{\hat{\sigma}_{n}^{2}}{\rho}, \\
	1-\frac{\ln(\rho k^{d}_{a}+\hat{\sigma}_{n}^{2})-\ln( \hat{\sigma}_{n}^{2})}{2\ln\rho} , & \rho\hat{\sigma}_{n}^{2} \geq k^{d}_{a}+\frac{\hat{\sigma}_{n}^{2}}{\rho},
	\end{cases}
	\end{aligned}
	\end{equation}
where $k^{d}_{a}=P_{a}G_{aw}L_{aw}^{d,\mathbb{B}}|{h}^{d,\mathbb{B}}_{aw}|^{2}$, $\rho$ and $\hat{\sigma}_{n}^{2}$ are the parameter that quantifies the size of the uncertainty and nominal noise power, respectively, which are defined in Section \ref{sec_noise_uncertainty}.
\end{lemma}

\begin{proof}
The proof is similar to Lemma \ref{2th1_NJ}, we omit it here.
\end{proof}

Similar to Theorem~\ref{2th2_NJ}, Alice and Bob only rely on the expected measure of $P_{ew}^{*,d}$ to evaluate the covertness under the DM mode. Note that $\mathbb{E}[P_{ew}^{*,d}]$ is related with the numerical integration of $|h^{m,\mathbb{B}}_{aw}|^{2}$. Thus before deriving $\mathbb{E}[P_{ew}^{*,d}]$, we first present the following lemma.
\begin{lemma}\label{lemma:calu}
If $h^{m,\mathbb{B}}_{aw}$ $(\mathbb{B}\in \{L,N\})$ is the channel fading coefficient of the mmWave channel, we have
\begin{align}\label{Eq:calu}
\int_{0}^{a}\!\!\!{x}f_{|h^{ m,\mathbb{B}}_{ aw}|^{2}}({x})d{x}\!=\!\!\sum_{r=1}^{S_{\mathbb{B} }}\!\binom{S_{\mathbb{B} }}{r}(-1)^{r}\!\!\left[ae^{-r\xi _{\mathbb{B}}a}-\frac{1\!-\!e^{-r\xi _{\mathbb{B}}a}}{r\xi _{\mathbb{B}}}\right],
\end{align}
where $f_{|h^{ m,\mathbb{B}}_{ aw}|^{2}}$ is the PDF of ${|h^{ m,\mathbb{B}}_{ aw}|^{2}}$.
\end{lemma}

\begin{proof}
The detailed proof is presented in Appendix~\ref{ap1}.
\end{proof}

Note that the probability of channel model uncertainty as well as antenna gain uncertainty needs to be considered simultaneously. Thus, when UAV adopts the DM transmission mode, $\mathbb{E}[P_{ew}^{*,d}]$ can be given as follows.
\begin{theorem} \label{th2}
When Alice adopts the DM transmission mode, the expected value $\mathbb{E}[P_{ew}^{*,d}]$ from Alice's perspective is given as (\ref{eqn1_NJ}),
\begin{figure*}[t]
\begin{equation} \label{eqn1_NJ}
\mathbb{E}[P_{ew}^{*,d}\!=\,\smashoperator{\sum_{\mathbb{A}\in\{M,S\}}}P_{w}^{\mathbb{A}}\;\smashoperator{\sum_{\mathbb{B}\in\{L,N\}}}P_{\!aw}^{\mathbb{B}}(1\!+\!\Theta^{d}_\mathbb{B})\!\left \{\!1\!\!-\!\frac{1}{2\ln\rho}\!\left\{\!\ln \!\left [\rho P_{a}G_{a}^{\mathbb{A}}G_{w}^{\mathbb{A}}L_{aw}^{m,\mathbb{B}}\!\!\left(\!\varrho^{d}\Theta^{d}_\mathbb{B}\!- \!\sum_{r=1}^{S_{\mathbb{B} }}\binom{\!S_{\mathbb{B}}}{r}(-1)^{r}\frac{1\!-\!e^{-r\xi _{\mathbb{B}}\varrho^{d}}}{r\xi _{\mathbb{B}}}\!\right)\!\!+\!\hat{\sigma}_{n}^{2} \right]\!\!-\!\ln(\hat{\sigma}_{n}^{2})\! \right\}\!\right\}.
\end{equation}
\rule[1cm]{\textwidth}{0.04em}
\vspace*{-0.6cm}
\end{figure*}
where $G_{a}^{\mathbb{A}}$ is refer to (\ref{neq1_gaw}), $\varrho^{d}=\frac{(\rho^{2}-1)\hat{\sigma}_{n}^{2}}{\rho P_{a}G_{a}^{\mathbb{A}}G_{w}^{\mathbb{A}}L_{aw}^{m,\mathbb{B}}}$, and $\Theta^{d}_\mathbb{B}$ is defined as
\begin{equation} \label{eqn2}
\Theta^{d}_\mathbb{B}=\sum\nolimits_{r=1}^{S_{\mathbb{B} }}\binom{S_{\mathbb{B} }}{r}(-1)^{r} e^{-r\xi _{\mathbb{B}}\varrho^{d}}.
\end{equation}
\end{theorem}

\begin{proof}
For convenience, let's denote $k_1^{d}=\rho\hat{\sigma}_{n}^{2}$ and $k_2^{d}=k^{d}_{a}+\frac{\hat{\sigma}_{n}^{2}}{\rho}$. According to (\ref{eq4_NJ}) in Lemma~\ref{th1_NJ}, we have
\begin{footnotesize}
\begin{align}  \label{eq8_NJ}
\mathbb{E}[&P_{\!ew}^{*,d}]\!=\!\mathbb{E}_{k_{1}^{d}< k_{2}^{d}}[P_{ew}^{*,d}]\mathbb{P}(k_{1}^{d}\!<\!k_{2}^{d})\! +\!\mathbb{E}_{k_{1}^{d}\!\geq \! k_{2}^{d}}[P_{ew}^{*,d}]\mathbb{P}(k_{1}^{d}\!\geq \!k_{2}^{d}) \notag \\
&=\mathbb{P}(k_{1}^{d}\geq k_{2}^{d})\notag\\
&\times\left( \!\!1\!\!-\!\frac{\! \ln\!\left( \!\rho P_{a}G_{a}^{\mathbb{A}}G_{w}^{\mathbb{A}}L_{aw}^{m,\mathbb{B}}\mathbb{E}_{k_{1}^{d}\geq k^{d}_{2}}\!\!\left[|{h}^{m,\mathbb{B}}_{aw}|^{2}\right]\!\!+\!\hat{\sigma}_{n}^{2}\!\right)\!\!-\!\ln(\!\hat{\sigma}_{n}^{2}\!)}{2\ln(\rho)}\!\!\right)\!\!,
\end{align}
\end{footnotesize}\normalsize
where the term $\mathbb{P}(k^{d}_{1}\geq k^{d}_{2})$, we can get
\begin{align} \label{eq9_NJ}
\mathbb{P}&(k^{d}_{1}\!\geq \!k^{d}_{2})\!=\mathbb{P}\left(\rho\hat{\sigma}_{n}^{2} \geq P_{a}G_{a}^{\mathbb{A}}G_{w}^{\mathbb{A}}L_{aw}^{m,\mathbb{B}}|{h}^{m,\mathbb{B}}_{ aw}|^{2}+\frac{\hat{\sigma}_{n}^{2}}{\rho}\right)\notag\\
&\!=\mathbb{P}\left (|{h}^{m,\mathbb{B}}_{aw}|^{2}\!\leq  \!\varrho^{d}\right) \!=\! \sum_{r=0}^{S_{\mathbb{B} }}\binom{S_{\mathbb{B} }}{r}(-1)^{r} e^{-r\xi _{\mathbb{B}}\varrho^{d}}.
\end{align}

For the expectation term in (\ref{eq8_NJ}), we have
\begin{align} \label{eq11_NJ}
\mathbb{E}&_{k_{1}^{d}\geq k^{d}_{2}} \left [ |{h}^{m,\mathbb{B}}_{aw}|^{2} \right]=\mathbb{E}\left [ |{h}^{m,\mathbb{B}}_{aw}|^{2}\Bigg||{h}^{m,\mathbb{B}}_{aw}|^{2}\leq \varrho^{d}\right ] \notag \\
&=\int_{0}^{ \varrho^{d}}{x}f_{|{h}^{m,\mathbb{B}}_{aw}|^{2}}({x})d{x} \notag \\
&\overset{(a)}=\sum_{r=1}^{S_{\mathbb{B} }}\binom{S_{\mathbb{B}}}{r}(-1)^{r}\left[ \varrho^{d}e^{-r\xi _{\mathbb{B}} \varrho^{d}}-\frac{1-e^{-r\xi _{\mathbb{B}} \varrho^{d}}}{r\xi _{\mathbb{B}}}\right],
\end{align}
where step (a) is due to Lemma~\ref{lemma:calu}. Finally, substituting (\ref{eq9_NJ}), (\ref{eq11_NJ}) into (\ref{eq8_NJ}), we can obtain $\mathbb{E}[P_{ew}^{*,d}]$ as (\ref{eqn1_NJ}).
\end{proof}

\subsection{Effective Covert Rate and Covert Shannon Capacity}

Similarly as in Section~\ref{oa_ecrcsc}, the ECR $R^{d}_{ab}$ of the system under the DM transmission mode can be given as the following theorem.
\begin{theorem}\label{le3_s1_out_NJ}
When UAV adopts the DM transmission mode, the ECR $R_{ab}^{d}$ of the concerned A2G system is given by
\begin{equation}
\label{Eq:s1_out_NJ}
R^{d}_{ab}\!=\!R_{b}\!\times\!\left\{\! 1\!-\!\smashoperator{\sum_{\mathbb{A}\in\{\! M,S\!\}}}{P}^{\mathbb{A}}_{b}\;\smashoperator{\sum_{\mathbb{B}\in\{\!L,N\!\}}}P_{ab}^{\mathbb{B}}\left[1\!+\! \mathcal{F}^{d}_{\mathrm{Ei}}(\!\rho\hat{\sigma}_{n}^{2}\!)\!-\!\mathcal{F}^{d}_{\mathrm{Ei}}\!\left(\!\frac{\hat{\sigma}_{n}^{2}}{\rho}\!\right)\!\right]\!\right\},
\end{equation}
where $R_{b}$ is the target covert rate, and $\mathcal{F}^{d}_{\mathrm{Ei}}(\cdot)$ in (\ref{Eq:s1_out_NJ}) is defined as following
\begin{align}\label{Eq:s2_fp}
 \mathcal{F}^{d}_{\mathrm{Ei}}(x)\!=\! \sum_{r=1}^{S_{\mathbb{B}}}\binom{S_{\mathbb{B} }}{r}(-1)^{r}\frac{1}{2\ln{\rho}} \mathrm{Ei}\left({\frac{-r\xi _{\mathbb{B}}\gamma_{th}x}{P_aG_{a}^{M}G_{b}^{\mathbb{A}}L_{ab}^{m,\mathbb{B}}}}\right).
\end{align}
\end{theorem}

\begin{proof}
Similarly, to analyze the ECR, we need to determine the outage probability $P^{d}_{out}$ at Bob. Due to the uncertainties of the LoS/NLoS channel and antenna gains, the outage probability at Bob $P^{d}_{out}$ under the DM transmission mode can be determined by
\begin{equation}\label{Eq:s1_proof_out11_NJ}
P^{d}_{out}\!=\!{\sum\nolimits_{\mathbb{A}\in\{M,S\}}}{P}^{\mathbb{A}}_{b}{\sum\nolimits_{\mathbb{B}\in\{L,N\}}}P_{ab}^{\mathbb{B}}\times \mathbb{P}\left(\gamma^{d}_{ab}<\gamma_{th}\right).
\end{equation}
According to (\ref{eq2_bob}), the SNR at Bob is given as $\gamma^{d} _{ab}={P_a G_{ab}L_{ab}^{d,\mathbb{B}}|h^{d,\mathbb{B}}_{ab}|^{2}}/{\sigma _{b}^{2}}$.
Thus, the outage probability at Bob $P^{d}_{out}$ is given by
\begin{align}\label{Eq:s2_proof_out_NJ}
\mathbb{P}&\left( \gamma^{d}_{ab}<\gamma_{th}\right)
=\mathbb{P}\left(\frac{P_aG_{a}^{M}G_{b}^{\mathbb{A}}L_{ab}^{m,\mathbb{B}}|{h}^{m,\mathbb{B}}_{ab}|^{2} }{\sigma _{b}^{2}}<\gamma_{th}\right) \notag \\
& =\mathbb{P}\left(|{h}^{m,\mathbb{B}}_{ab}|^{2}<\frac{\sigma _{b}^{2}\gamma_{th}}{P_aG_{a}^{M}G_{b}^{\mathbb{A}}L_{ab}^{m,\mathbb{B}}}\right)   \notag \\
&=\sum\nolimits_{r=0}^{S_{\mathbb{B} }}\binom{S_{\mathbb{B} }}{r}(-1)^{r}\mathbb{E}_{\sigma _{b}^{2}} \left [ \exp \left({\frac{-r\xi _{\mathbb{B}}\gamma_{th}\sigma _{b}^{2}}{P_aG_{a}^{M}G_{b}^{\mathbb{A}}L_{ab}^{m,\mathbb{B}}}}\right)\right ]\notag \\
&=1\!+\!\sum_{r=1}^{S_{\mathbb{B} }}\!\binom{S_{\mathbb{B} }}{r}\!(-1)^{r}\!\!\int_{\frac{\hat{\sigma}_{n}^{2}}{\rho}}^{\rho\hat{\sigma}_{n}^{2}}\!\!\exp \!\left(\!{\frac{-r\xi _{\mathbb{B}}\gamma_{th}x}{P_aG_{a}^{M}G_{b}^{\mathbb{A}}L_{ab}^{m,\mathbb{B}}}}\!\right)\!\frac{1}{2x\ln{\rho} }dx \notag\\
&\overset{(a)}{=}1+\sum\nolimits_{r=1}^{S_{\mathbb{B} }}\binom{S_{\mathbb{B} }}{r}(-1)^{r}\frac{1}{2\ln{\rho} }\mathrm{Ei}\left({\frac{-r\xi _{\mathbb{B}}\gamma_{th}x}{P_aG_{a}^{M}G_{b}^{\mathbb{A}}L_{ab}^{m,\mathbb{B}}}}\right)\Bigg|_{\frac{\hat{\sigma}_{n}^{2}}{\rho}}^{\rho\hat{\sigma}_{n}^{2}} \notag\\
&=1+\sum\nolimits_{r=1}^{S_{\mathbb{B} }}\binom{S_{\mathbb{B} }}{r}(-1)^{r}\!\frac{1}{2\ln{\rho} }\mathrm{Ei}\left({\frac{-r\xi _{\mathbb{B}}\gamma_{th}\rho\hat{\sigma}_{n}^{2}}{P_aG_{a}^{M}G_{b}^{\mathbb{A}}L_{ab}^{m,\mathbb{B}}}}\right)  \notag \\
&  -\sum\nolimits_{r=1}^{S_{\mathbb{B} }}\!\binom{S_{\mathbb{B} }}{r}(-1)^{r}\frac{1}{2\ln{\rho} }\mathrm{Ei}\left({\frac{-r\xi _{\mathbb{B}}\gamma_{th}\hat{\sigma}_{n}^{2}}{P_aG_{a}^{M}G_{b}^{\mathbb{A}}L_{ab}^{m,\mathbb{B}}\rho}}\right),
\end{align}
where step (a) is due to $\int e^{ax}x^{-1}dy=\mathrm{Ei}(ax)$ \cite[Eq. (2.325.1)]{gradshteyn2014table}. Then, submitting (\ref{Eq:s2_proof_out_NJ}) into (\ref{Eq:s1_proof_out11_NJ}), we obtain (\ref{Eq:s1_out_NJ}).
\end{proof}

Next, we derive the covert Shannon capacity $C^{d}_{ab}$ of the considered system under the DM transmission mode.
\begin{theorem} \label{th_s1_covert_rate_NJ}
In the considered A2G system, the CSC $C^{d}_{ab}$ of the system under the DM transmission mode is determined by
\begin{equation}
\label{s1_rab_NJ}
C^{d}_{ab}\!=\,\,\smashoperator{\sum_{\mathbb{A}\in\{M,S\}}}P^{\mathbb{A}}_{b}\,\,\smashoperator{\sum_{\mathbb{B}\in\{L,N\}}}\! P_{ab}^{\mathbb{B}} \frac{W^{d}}{2\ln2\!\times\!\ln\rho}\!\left [\! \mathcal{F}^{d}_{\mathbf{Li}}\!\left (\!\frac{1}{\rho\hat{\sigma}_{n}^{2}}\!\right)\!-\! \mathcal{F}^{d}_{\mathbf{Li}}\!\left (\!\frac{\rho}{\hat{\sigma}_{n}^{2}}\!\right)\!\right ],
\end{equation}
where $W^{d}$ is the bandwidth of mmWave, and $\mathcal{F}^{d}_{\mathbf{Li}}(\cdot)$ in (\ref{s1_rab_NJ}) is defined as following
\begin{align}\label{Eq:s2_fr}
\mathcal{F}^{d}_{\mathbf{Li}}(x)= \int_{0}^{\infty}\mathbf{Li}_{2} (-P_aG_{a}^{M}G_{b}^{\mathbb{A}}L_{ab}^{m,\mathbb{B}}xy)f_{|{h}^{m,\mathbb{B}}_{ab}|^{2}}(y)dy,
\end{align}
$\mathbf{Li}_{2}(\cdot)$ is refer to Theorem \ref{th_s2_covert_rate_TANJ}, and $f_{|h^{ m,\mathcal{B}}_{ab}|^{2}}(y)$ is refer to (\ref{eq_mmwave_pdf}).
\end{theorem}
\begin{proof}
The proof is similar to Theorem \ref{th_s2_covert_rate_TANJ}, we omit here.
\end{proof}


\section{Performance Optimization and Mode Selection}\label{sec_results}
In this section, we establish the optimization problems of the ECR and CSC maximization under the OM and DM transmission modes, and further propose a hybrid OM/DM transmission mode.

\subsection{Maximum Effective Covert Rate} \label{oecr}
From Theorem \ref{th3_s2_out_NJ} and \ref{le3_s1_out_NJ}, we note that the ECR $R_{ab}$ is related with the transmission power $P_{a}$ and target rate $R_b$. Besides, a large $P_{a}$ and $R_b$ results in a small $\mathbb{E}[P_{ew}^{*,\mathbb{C}}]$ and a large $P_{out}$, respectively. Thus, for a certain position $U(x,y,z)$, the UAV intends to find its optimal transmission power $P_{a}$ and optimal target rate $R_b$ to maximize $R_{ab}$ with the DEP constraint. The corresponding optimal problems under the OM and DM transmission modes can be respectively formulated as the follows
\begin{subequations} \label{Eq:Max_Coverate_om}
\begin{align}
\bar{R}_{a,b}^{*,o}(x_a,y_a,h_a)&=\mathop{\max}\limits_{P_{a},R_{b}} R_{ab}^{o}(x_a,y_a,h_a), \\
\qquad s.t.\, \,\, &\mathbb{E}[P_{ew}^{*,\mathbb{C}}] \geq 1-\epsilon \label{Eq:con_dep_r_om} \\
& P_{a} \le P_{max} \label{Eq:con_pa_r_om}
\end{align}
\end{subequations}
and
\begin{subequations} \label{Eq:Max_Coverate_dm}
\begin{align}
\bar{R}_{a,b}^{*,d}(x_a,y_a,h_a)&=\mathop{\max}\limits_{P_{a},R_{b}} R_{ab}^{d}(x_a,y_a,h_a), \\
\qquad s.t.\, \,\, &\mathbb{E}[P_{ew}^{*,\mathbb{C}}] \geq 1-\epsilon \label{Eq:con_dep_r_dm} \\
& P_{a} \le P_{max} \label{Eq:con_pa_r_dm}
\end{align}
\end{subequations}
where $\bar{R}^{*,o}_{ab}(x,y,h_a)$ (resp. $\bar{R}^{*,d}_{ab}(x,y,h_a)$) denotes the maximum $R_{ab}^{o}(x,y,h_a)$ (resp. $R_{ab}^{d}(x,y,h_a)$), which characterizes the maximum of average \emph{successfully} transmitted amount of information subject to a covertness requirement $\epsilon$. Although we cannot obtain a closed-form result of optimal transmission power and optimal target rate due to transcendental functions, we can obtain the optimal solutions by numerical search methods.

\subsection{Maximum Covert Shannon Capacity}\label{ocsc}

According to Theorem \ref{th_s2_covert_rate_TANJ} and \ref{th_s1_covert_rate_NJ}, we know that a larger $P_{a}$ will lead to a larger CSC $C_{ab}$. But, according to the detection strategy of Willie, a large $P_{a}$ also lead to a lower DEP $\mathbb{E}[P_{ew}^{*,\mathbb{C}}]$. Therefore, UAV hopes to maximize the CSC of the A2G system by optimizing the transmission power under the OM and DM transmission modes, which can be respectively formulated as
\begin{subequations} \label{Eq:Max_capcacity_om}
\begin{align}
\bar{C}_{a,b}^{*,o}(x_a,y_a,h_a)&=\mathop{\max}\limits_{P_{a}} C_{ab}^{o}(x_a,y_a,h_a), \\
\qquad s.t.\, \,\, &\mathbb{E}[P_{ew}^{*,\mathbb{C}}] \geq 1-\epsilon  \label{Eq:con_dep_c_om}  \\
& P_{a} \le P_{max} \label{Eq:con_pa_c_om}
\end{align}
\end{subequations}
and
\begin{subequations} \label{Eq:Max_capcacity_dm}
\begin{align}
\bar{C}_{a,b}^{*,d}(x_a,y_a,h_a)&=\mathop{\max}\limits_{P_{a}} C_{ab}^{d}(x_a,y_a,h_a), \\
\qquad s.t.\, \,\, &\mathbb{E}[P_{ew}^{*,\mathbb{C}}] \geq 1-\epsilon  \label{Eq:con_dep_c_dm}  \\
& P_{a} \le P_{max} \label{Eq:con_pa_c_dm}
\end{align}
\end{subequations}
where $\bar{C}^{*,o}_{ab}(x_a,y_a,h_a)$ (resp. $\bar{C}^{*,d}_{ab}(x_a,y_a,h_a)$) denote the maximum $C_{ab}^{o}(x_a,y_a,h_a)$ (resp. $C_{ab}^{d}(x_a,y_a,h_a)$), which characterizes the \emph{upper bound} of the time average rate of messages delivered from transmitter to the destination subject to a covertness requirement $\epsilon$. Similarly, we can solve these optimization problems through numerical search methods.

\subsection{Optimal Transmission Mode for Covert Communication}

Although mmWave has a larger bandwidth than low-frequency microwave, it also have the disadvantage of faster attenuation. Thus, with the dynamic change of the distance from UAV to Bob and Willie, respectively, the hybrid OM/DM transmission mode of the UAV would be superior to the pure OM or DM transmission mode in terms of covert performance. To confirm our idea, we propose a hybrid OM/DM transmission mode which allows UAV to adaptively switch between the OM and DM transmission modes based on (\ref{Eq:Max_Coverate_om})-(\ref{Eq:Max_capcacity_dm}). Specifically, for a given position ($x_a,y_a,h_a$) of UAV, we first calculate the distance $d_{ab}$ from UAV to Bob and the distance $d_{aw}$ from UAV to Willie. Then, we submit them into (\ref{Eq:Max_Coverate_om}) and (\ref{Eq:Max_Coverate_dm}) (resp. (\ref{Eq:Max_capcacity_om}) and (\ref{Eq:Max_capcacity_om})) and solve the optimal problems. We select the optimal transmission mode $I_{ECR}$ (resp. $I_{CSC}$) by comparing $\bar{R}_{ab}^{*,o}(x_a,y_a,h_a)$ and $\bar{R}_{ab}^{*,d}(x_a,y_a,h_a)$ (resp. $\bar{C}_{a,b}^{*,o}(x_a,y_a,h_a)$ and $\bar{C}_{a,b}^{*,d}(x,y,z)$), where $I_{ECR}$ and $I_{CSC}$ are the indicators of the optimal selection mode for maximizing the ECR and CSC, which are respectively denoted as
	\begin{equation} \label{eq:I_ecr}
	\begin{aligned}
	I_{ECR}=\begin{cases}
	 \text{OM mode},& \bar{R}_{ab}^{*,o}(x_a,y_a,h_a)\ge \bar{R}_{ab}^{*,d}(x_a,y_a,h_a),\\
	 \text{DM mode},& otherwise,
	\end{cases}
	\end{aligned}
	\end{equation}
and
	\begin{equation} \label{eq:I_csc}
	\begin{aligned}
	I_{CSC}=\begin{cases}
	 \text{OM mode},& \bar{C}_{ab}^{*,o}(x_a,y_a,h_a)\ge \bar{C}_{ab}^{*,d}(x_a,y_a,h_a),\\
	 \text{DM mode},& otherwise.	
 \end{cases}
	\end{aligned}
	\end{equation}
 Overall, the hybrid OM/DM transmission mode selection algorithm for optimal covert communication can be described as Algorithm~\ref{Alg:Selection_Max_CR}.
\begin{algorithm}\label{Alg:Selection_Max_CR}
    \caption{Hybrid OM/DM Transmission Mode Selection Algorithm}
    \SetAlgoLined
    \SetKwInOut{Input}{input}
    \SetKwInOut{Output}{output}
    \DontPrintSemicolon
    \Input{UAV's location $U(x_a,y_a,h_a)$, Bob's location $U(x_b,y_b,h_b)$, Willie's location $U(x_w,y_w,h_w)$.}
    \Output{The mode selection indicator $I_{ECR}$ or $I_{CSC}$.}
       Calculate the distances $d_{ab}$ and $d_{aw}$, respectively.\\
        \uIf {UAV intends to maximize the ECR}
        {Submit $d_{ab}$ and $d_{aw}$ into (\ref{Eq:Max_Coverate_om}) and (\ref{Eq:Max_Coverate_dm}) and solve the optimal problems;\\
         Obtain $I_{ECR}$ by comparing $\bar{R}_{ab}^{*,o}(x_a,y_a,h_a)$ with $\bar{R}_{ab}^{*,d}(x_a,y_a,h_a)$ according to (\ref{eq:I_ecr}).}
        \ElseIf{UAV intends to maximize the CSC}
        {Submit $d_{ab}$ and $d_{aw}$ into (\ref{Eq:Max_capcacity_om}) and (\ref{Eq:Max_capcacity_dm}) and solve the optimal problems;\\
         Obtain $I_{CSC}$ by comparing $\bar{C}_{ab}^{*,o}(x_a,y_a,h_a)$ with $\bar{C}_{ab}^{*,d}(x_a,y_a,h_a)$ according to(\ref{eq:I_csc}).}
         Return $I_{ECR}$ or $I_{CSC}$.
\end{algorithm}


\section{Numerical Results} \label{sec_num_results}

This section provides extensive numerical results to illustrate the system performance under both the OM mode and DM mode, such that optimal model for covert communication in the concerned hybrid microwave/mmWave A2G systems can be obtained.

In the following, we considered the A2G system where Bob is located at $(-500, 0, 0)$, Willie is located at $(1000, 0, 0)$, UAV flies at a fixed altitude of $h_{a}=500$m, the minimum and maximum safe distance limits for UAV flight are $d_{aw}^{min}=300$m $d_{aw}^{max}=1500$m, respectively. The OM transmission mode of Alice transmits at a typical frequency $2.5$GHZ with $40$-MHz bandwidth as \cite{akdeniz2014millimeter}, and the DM transmission mode transmits at a typical mmWave frequency $73$GHz and with $100$-MHz bandwidth as \cite{liu2017millimeter,zhang2020optimized}. We summarized the considered parameters in Table \ref{num_parameters}, unless explicitly mentioned.

\renewcommand{\arraystretch}{1.4}
\begin{table}[t]
	\centering
	\caption{Network Parameter Settings}
	\label{num_parameters}
	\begin{tabular}{|m{0.30\textwidth}|m{0.13\textwidth}|}
		\hline
\textbf{Network Parameters} & \textbf{Values} \\
		\hline
\hline
		UPA antenna elements ($\mathcal{N}_a$, $\mathcal{N}_b$, $\mathcal{N}_w$) & (6, 18, 18)\\
		\hline
		S-curve parameters ($\sigma$, $f$) & (4.88, 0.429) \\
		\hline
		OM channel path loss coefficients ($\beta_{L}^{o}$, $\beta_{N}^{o}$) & ($10^{-6}$, $10^{-7}$)\\
		\hline
		OM channel path loss exponents ($\alpha_{L}^{o}$, $\alpha_{N}^{o}$) &  (1.64, 2.71) \\
		\hline
		DM channel path loss coefficients ($\beta_{L}^{d}$, $\beta_{N}^{d}$) & ($10^{-6.11}$, $10^{-7.18}$)\\
		\hline
		DM channel path loss exponents ($\alpha_{L}^{d}$, $\alpha_{N}^{d}$) &  (2, 3) \\
		\hline
		Rician factor ($k_{0}$, $k_{\pi/2}$) & (5, 15) dBm \\
		\hline
		Nakagami-m fading shape parameters ($S_L$, $S_N$) & (3, 2) \\
		\hline
		Nominal noise power $\hat{\sigma} _{n}^{2}$ & -80 dBm\\
		\hline
		Noise uncertainty level $\rho$ & 2 dB\\
		\hline
		Target rate $R_b$ & 1 Mbits/s/Hz\\
		\hline
		Covertness requirement $\epsilon$ & 0.2\\
		\hline
	\end{tabular}
\end{table}

\subsection{Analysis of Expected Minimum DEP}

\begin{figure}[t]
    \centering
    \begin{minipage}{0.48\linewidth}
        \centering
        \includegraphics[width=\linewidth]{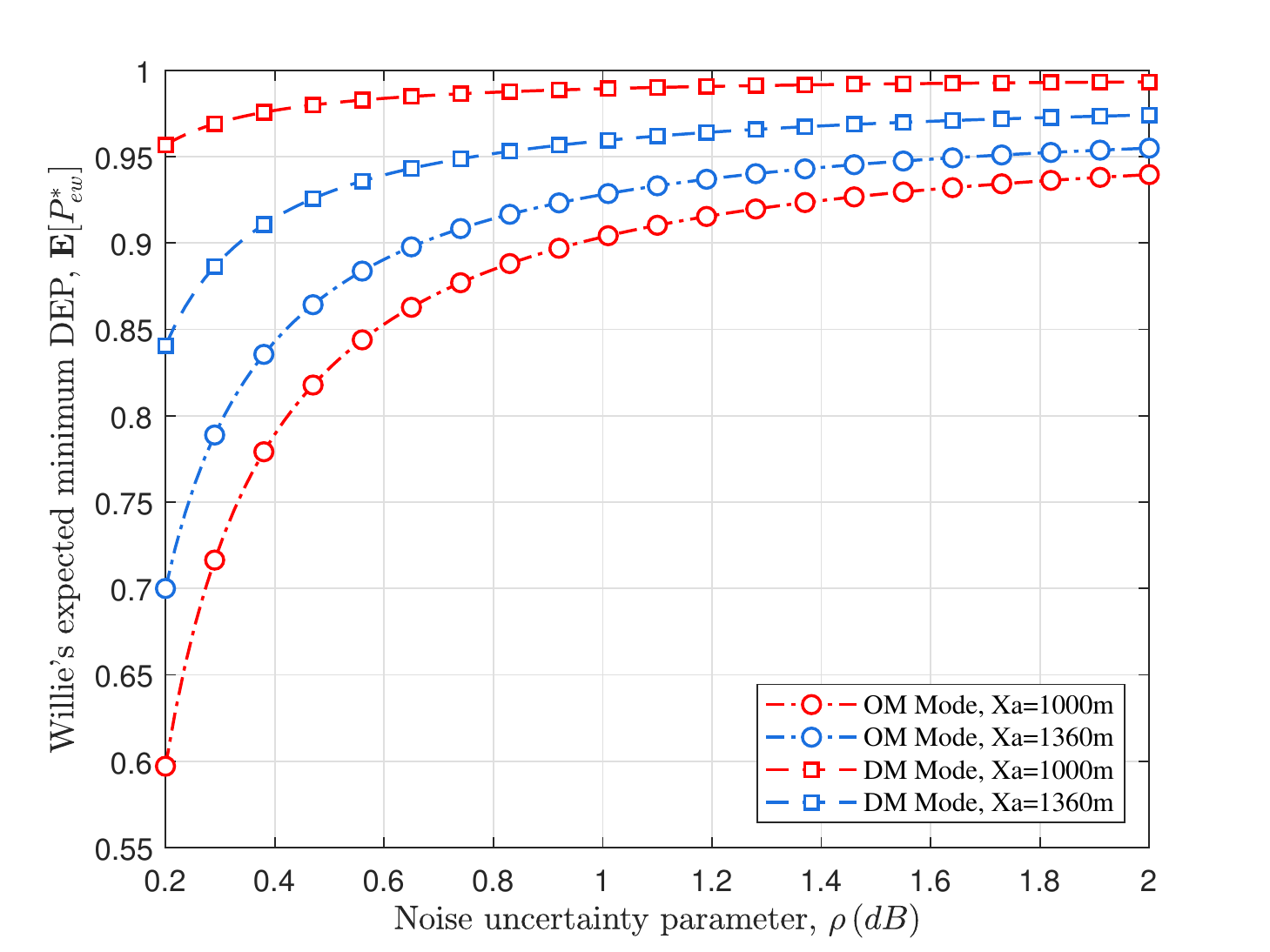}
        \vspace{-0.3cm}
        \caption{The expected minimum DEP $\mathbb{E}[P_{ew}^{*}]$ vs. noise uncertainty $\rho$.} 
    \label{fig4_1}
    \end{minipage}\hfill
    \begin{minipage}{0.48\linewidth}
        \centering
        \includegraphics[width=\linewidth]{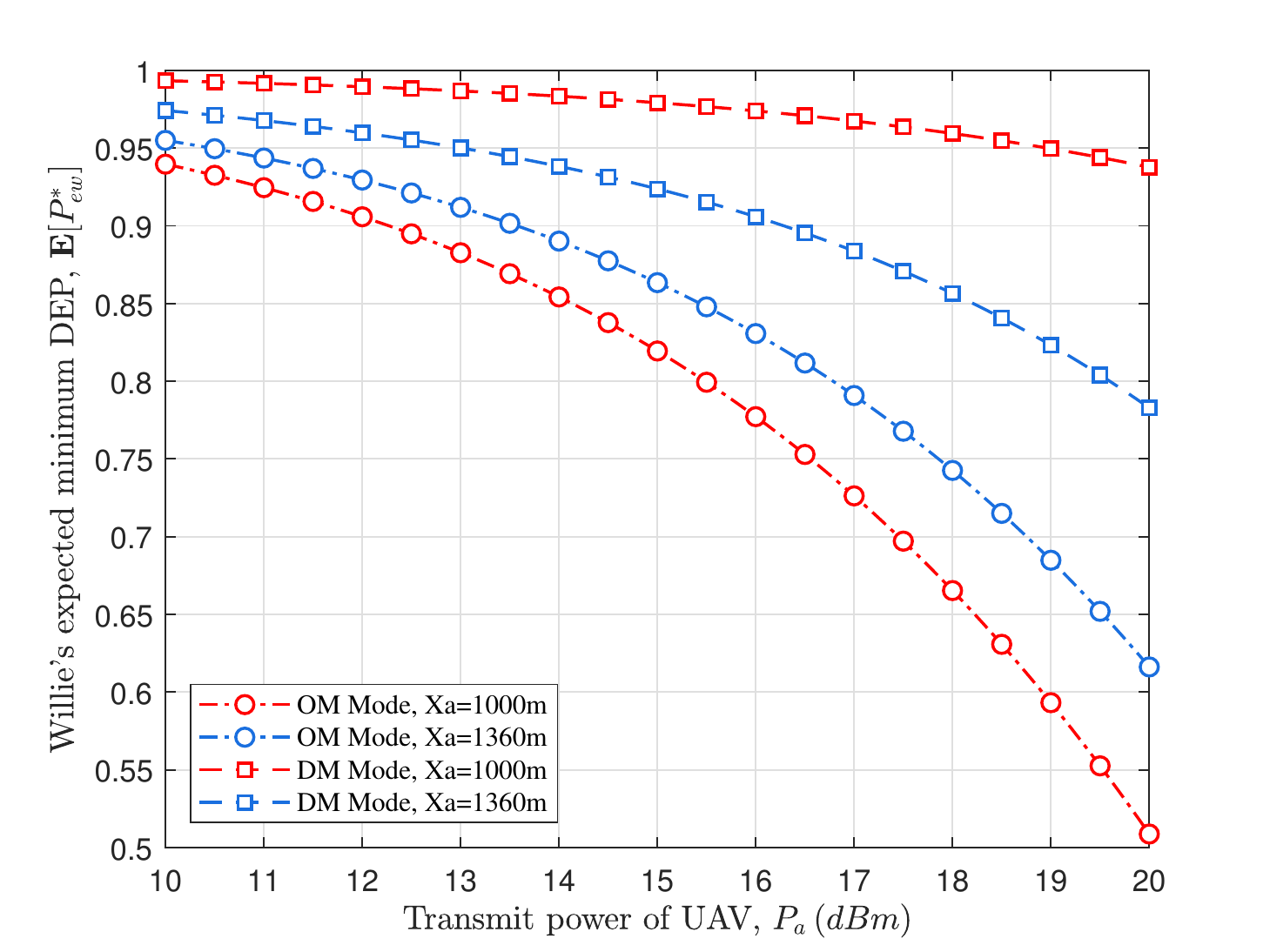}
        \vspace{-0.3cm}
        \caption{The expected minimum DEP $\mathbb{E}[P_{ew}^{*}]$ vs. transmission power $P_a$.} 
    \label{fig4_2}
    \end{minipage}
     \vspace{-0.5cm}
\end{figure}

\begin{figure}[t]
    \centering
    \begin{minipage}{0.48\linewidth}
        \centering
        \includegraphics[width=\linewidth]{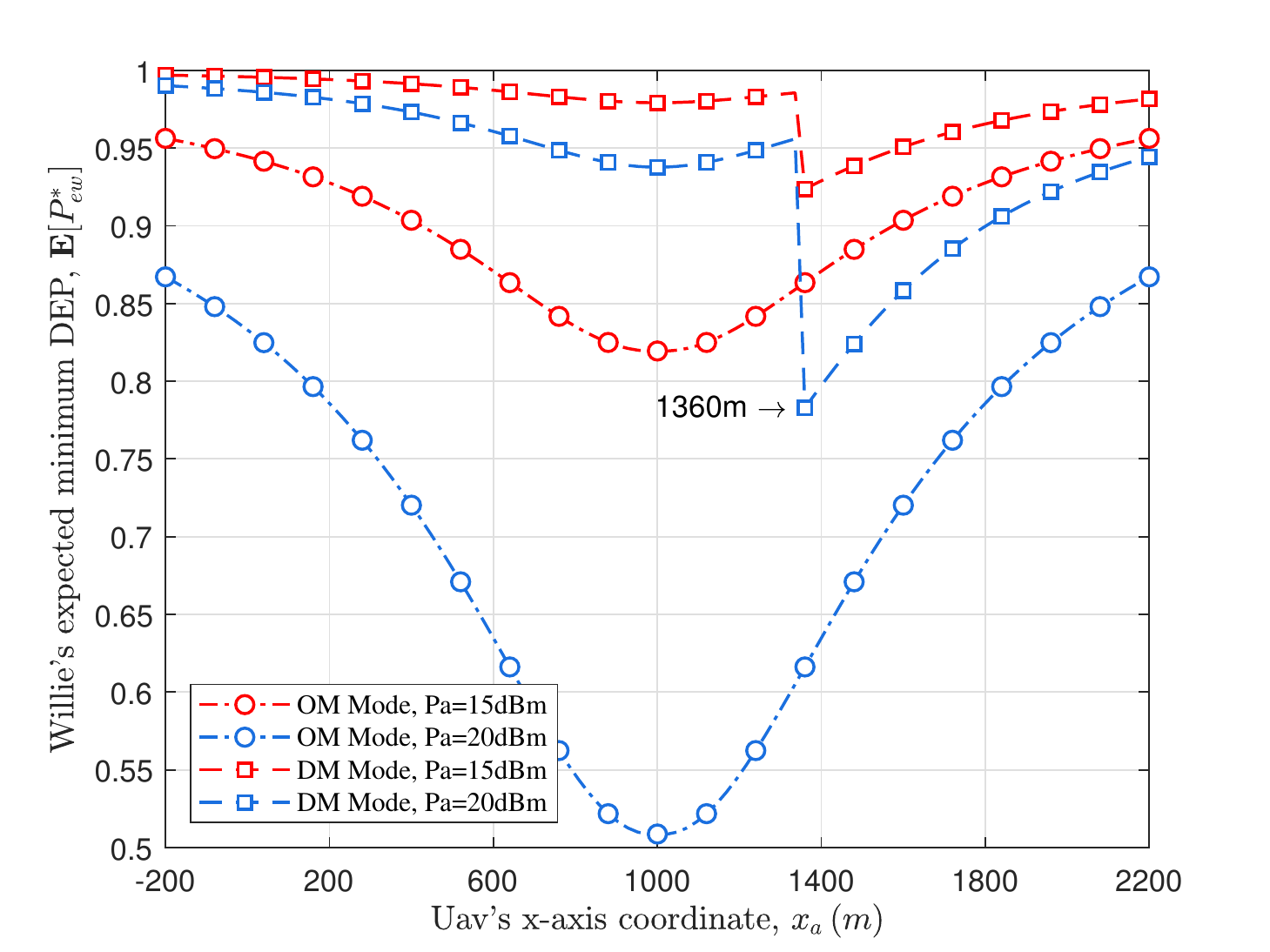}
        \vspace{-0.3cm}
       \caption{The expected minimum DEP $\mathbb{E}[P_{ew}^{*}]$ vs. horizontal position $x_a$.} 
    \label{fig4_3}
    \end{minipage}\hfill
    \begin{minipage}{0.48\linewidth}
        \centering
        \includegraphics[width=\linewidth]{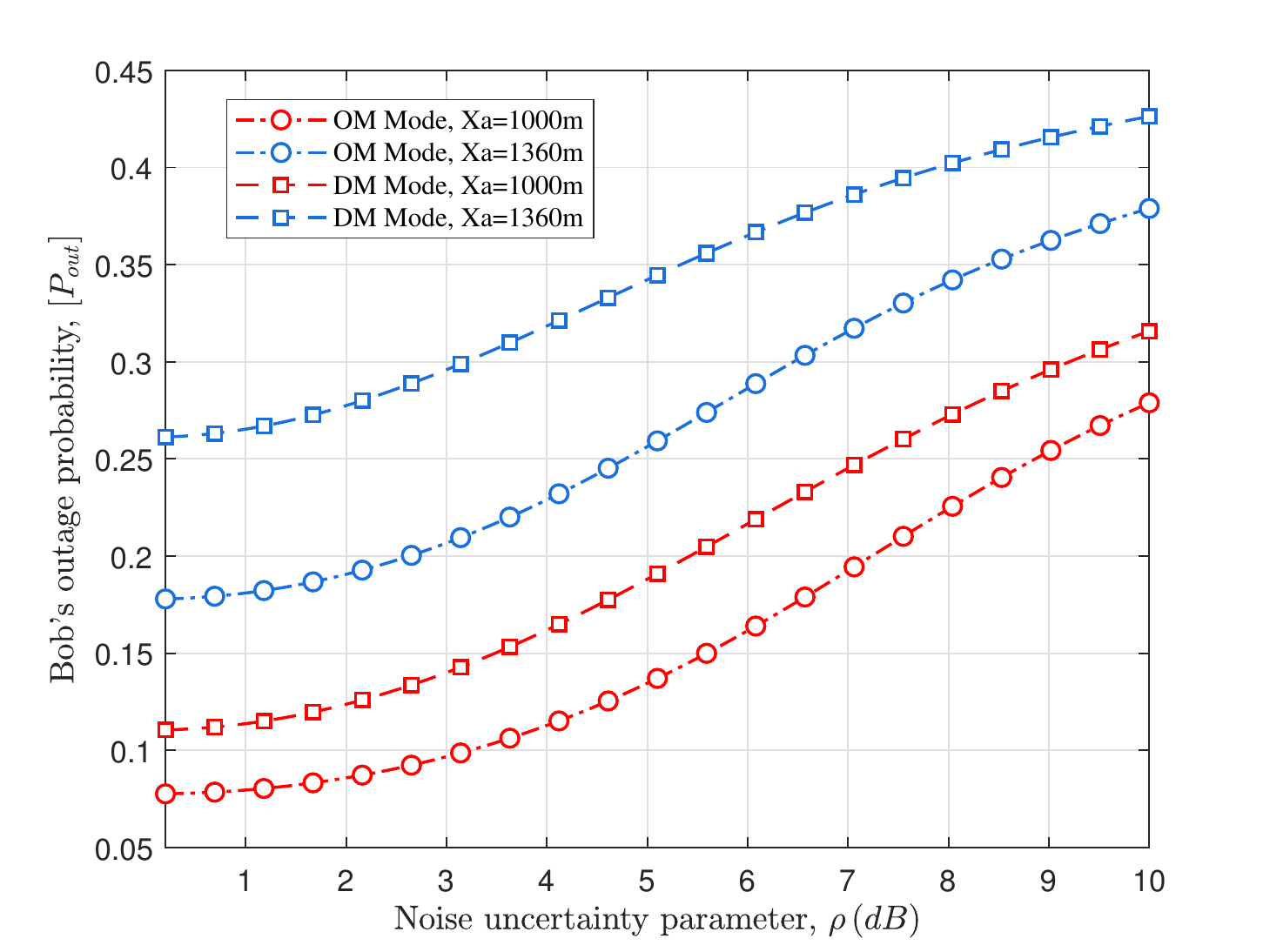}
        \vspace{-0.3cm}
        \caption{Outage probability $P_{out}$ vs. noise uncertainty parameter $\rho$.} 
      \label{fig5_1}
    \end{minipage}
     \vspace{-0.5cm}
\end{figure}

We first explore the impact of parameter $\rho$ which is used to quantify the level of the noise uncertainty on the expected minimum DEP $\mathbb{E}[P_{ew}^{*}]$ at Willie. As depicted in Fig.~\ref{fig4_1}, $\mathbb{E}[P_{ew}^{*}]$ keeps increasing as $\rho$ increases. This is because the greater the level of noise uncertainty, the more conducive to hide the message transmission of Alice. We also can observe that, for a given position of Alice, i.e., $x_a=1000m$ and $x_a=1360$m, the DM transmission mode is superior to the OM transmission mode in terms of covert performance. Due to the fast fading characteristic of mmWave, the probability of misjudgment under the DM mode is higher than that in the OM mode. From Fig.~\ref{fig4_1}, we can further observe that, for the DM transmission mode, $\mathbb{E}[P_{ew}^{*}]$ of Willie is higher when $x_a=1000m$ than that one when $x_a=1360m$, which is opposite of the OM transmission mode. The reason behind the phenomena can be explained that, Alice always steers her main lobe of the DM antenna to Bob, when $x_a=1360m$, Willie is also located in the boresight direction of the main lobe, such that he can obtain more power from Alice, thereby reducing $\mathbb{E}[P_{ew}^{*}]$.

We then investigate the impact of the transmission power of Alice $P_a$ on the expected minimum DEP $\mathbb{E}[P_{ew}^{*}]$ with $\rho=2$dB and $x_a=\{1000m, 1360m\}$. As shown in Fig.~\ref{fig4_2}, it can be observed that $\mathbb{E}[P_{ew}^{*}]$ is monotonically decreasing with respect to $P_a$. That is because as the transmission power increases, it is not conducive to hiding the information in the background noise. Also, as $P_a$ increases, the expected minimum DEP performance in OM transmission mode degrades faster than in DM transmission mode. Because the mmWave attenuates faster than the microwave, and the received power at Willie remains lower under the same transmission power.

To explore the impact of the position of Alice $x_a$ on the expected minimum DEP $\mathbb{E}[P_{ew}^{*}]$, we summarize in Fig.~\ref{fig4_3} to show how $\mathbb{E}[P_{ew}^{*}]$ varies with $x_a$ with $P_a=\{15\text{dBm}, 20\text{dBm}\}$. From Fig.~\ref{fig4_3}, we can see that $\mathbb{E}[P_{ew}^{*}]$ first decreases as $x_a$ increases from $-200$m to $1000$m, and then increases after $1000$m under the OM transmission mode. It is due to the fact that Alice just moves right above Willie, the distance between them is the closest, and Willie can detect the transmission behavior of Alice more accurately. For the DM transmission mode, we can see that when $x_a=1360m$, $\mathbb{E}[P_{ew}^{*}]$ decreases sharply to the minimal value and then gradually increase as $x_a$ increases. This interesting phenomenon can be explained as when $x_a=1360$m, Willie is just in the half-power beam width direction of the antenna array of link Alice $\to$ Bob and he can obtain more power benefiting from the main lobe antenna gain of Alice. But as $x_a$ increases, the distance between Willie and Alice will increase and the received power at Willie will decrease causing an increasing $\mathbb{E}[P_{ew}^{*}]$.

\subsection{Analysis of Outage Probability}

\begin{figure}[t]
    \centering
    \begin{minipage}{0.48\linewidth}
        \centering
        \includegraphics[width=\linewidth]{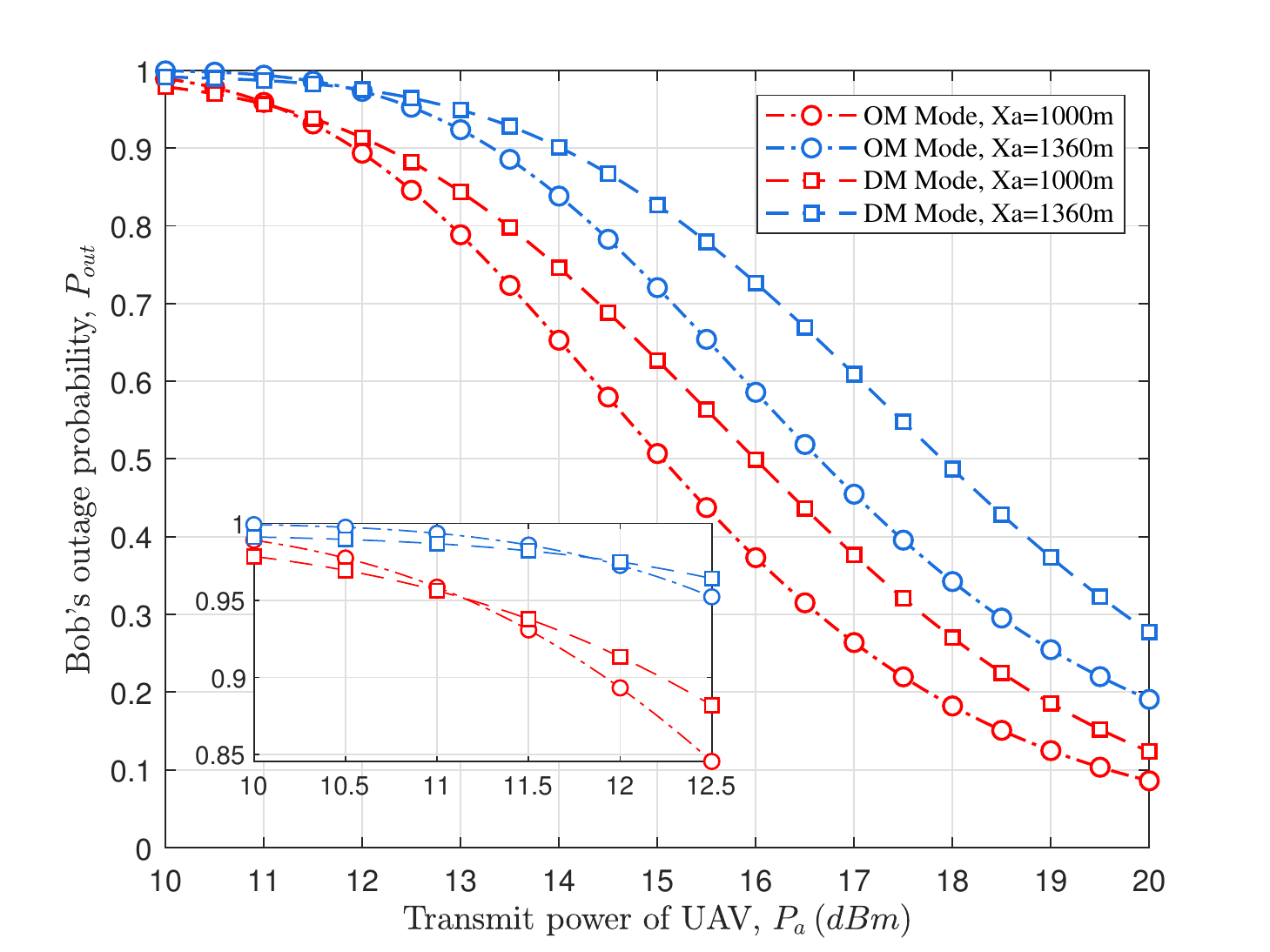}
        \vspace{-0.3cm}
        \caption{Outage probability $P_{out}$ vs. transmission power $P_a$.} 
    \label{fig5_2}
    \end{minipage}\hfill
    \begin{minipage}{0.48\linewidth}
        \centering
        \includegraphics[width=\linewidth]{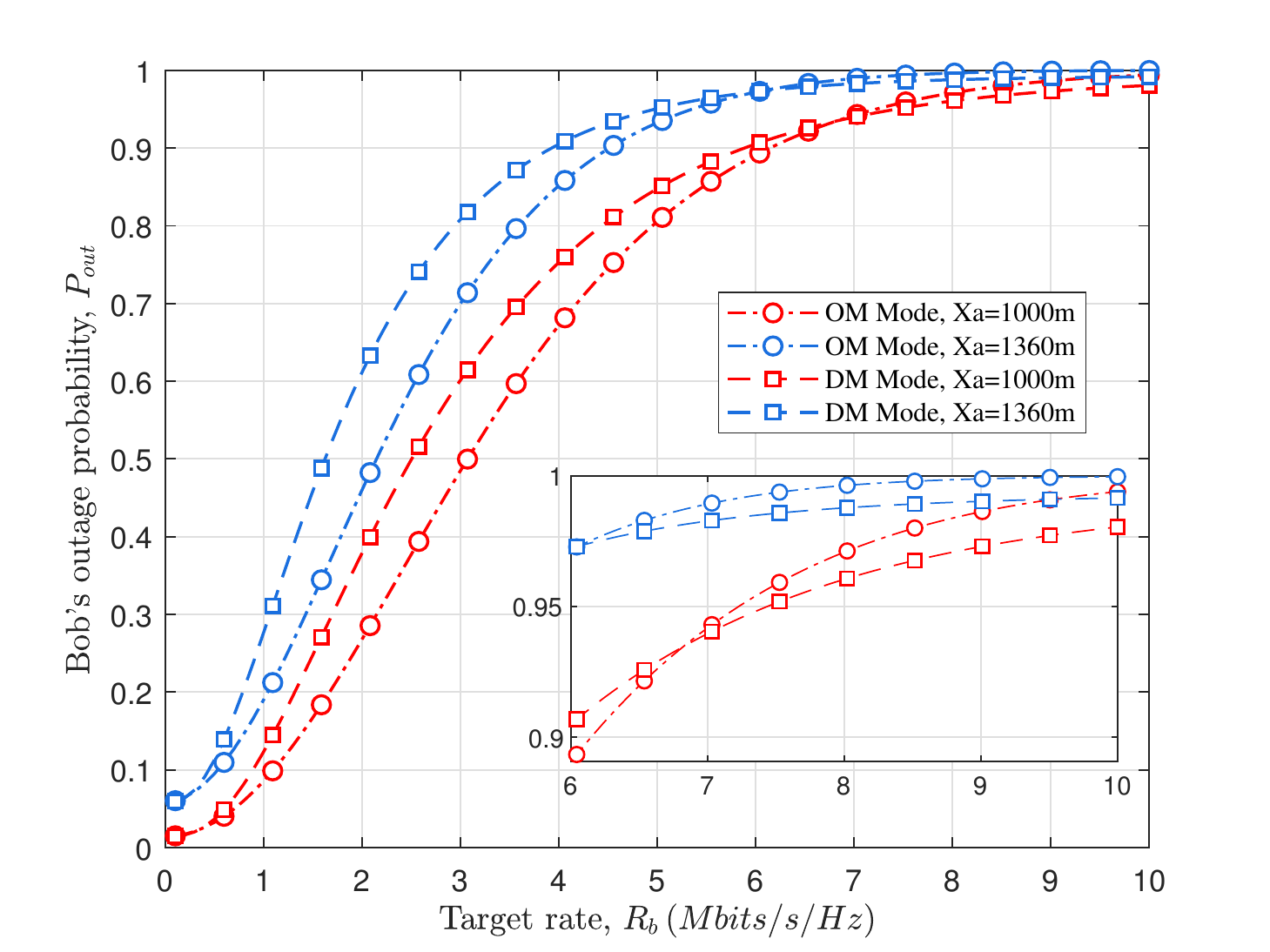}
        \vspace{-0.3cm}
        \caption{Outage probability $P_{out}$ vs. target rate $R_b$. }
    \label{fig5_3}
    \end{minipage}
    \vspace{-0.5cm}
\end{figure}

We plot Fig.~\ref{fig5_1} to explore the impact of noise uncertainty parameter $\rho$ on the outage probability $P_{out}$ under both the OM and DM transmission modes. From Fig.~\ref{fig5_1}, we can see that as $\rho$ increases, $P_{out}$ gradually increases. A larger $\rho$ leads to a lower SNR at Bob such that the outage probability is higher. Note that we also can find that under a given $x_a$ and $\rho$, the OM transmission mode outperforms the DM transmission mode in terms of the outage probability performance. This phenomenon can be attributed to the fast fading of mmWave resulting in a small SNR at Bob.

Fig.~\ref{fig5_2} is used to study the impact of transmission power $P_{a}$ on the outage probability $P_{out}$. We can see from Fig.~\ref{fig5_2} that the outage probability $P_{out}$ monotonically decreases with respect to $P_a$. Obviously, a larger $P_a$ will lead to a higher SNR resulting in a lower outage probability. A further careful observation of Fig.~\ref{fig5_2} indicates that when Alice transmits with a small $P_a$, $P_{out}$ under the OM transmission mode is higher than the one under the DM transmission mode, but once $P_a$ exceeds a certain value, $P_{out}$ under the OM transmission mode is lower than the one under the DM transmission mode. Note that when Alice transmits with a lower power, mmWave directional antenna can offset a tiny part of channel fading. Thus, the outage probability under the DM transmission mode is lower than that one under the OM transmission mode. But, the channel fading of mmWave is much faster than the microwave, such that when Alice transmits with a big power, the received power of Bob under the DM transmission mode than the one under the OM transmission mode, resulting in a higher outage probability.

To further study the changing trend of the outage probability $P_{out}$ varying with the target rate $R_b$ under both transmission modes, we draw Fig.~\ref{fig5_3}. We can observe that as the target rate $R_b$ increases, the outage probability $P_{out}$ also increases. It can be easily explained based on the definition of the outage probability. Besides, it also implies that there is a trade-off relationship between $P_{out}$ and $R_b$ which validates the rationality of optimizing $R_b$ in this work.

\subsection{Analysis of ECR and CSC}

\begin{figure} [tb]
    \centering
		\subfigure[ECR $\bar{R}^{*}_{ab}$]{
		\includegraphics[width=0.22\textwidth]{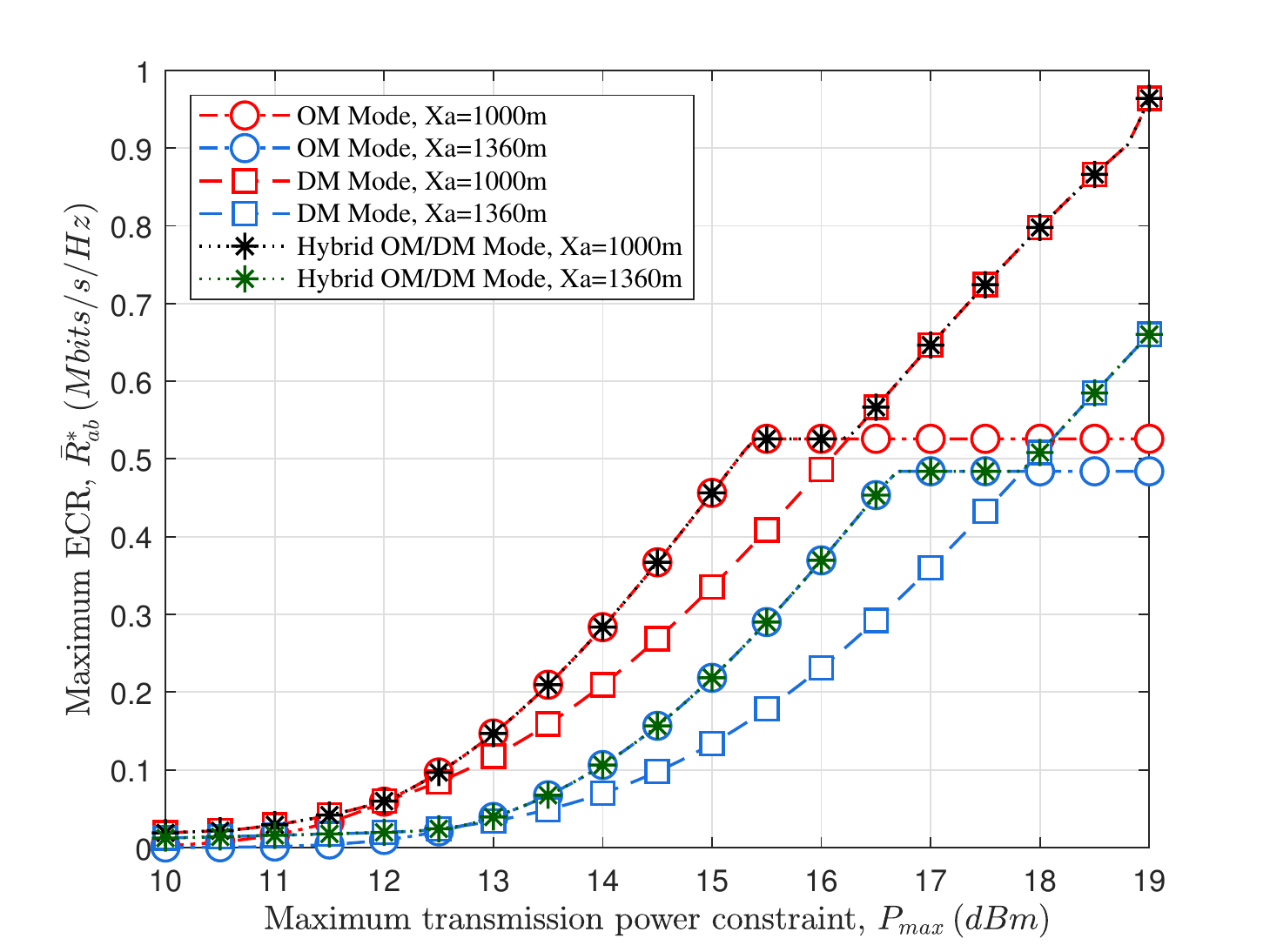} \label{fig61}
		}
		\subfigure[CSC $\bar{C}^{*}_{ab}$]{
		\includegraphics[width=0.22\textwidth]{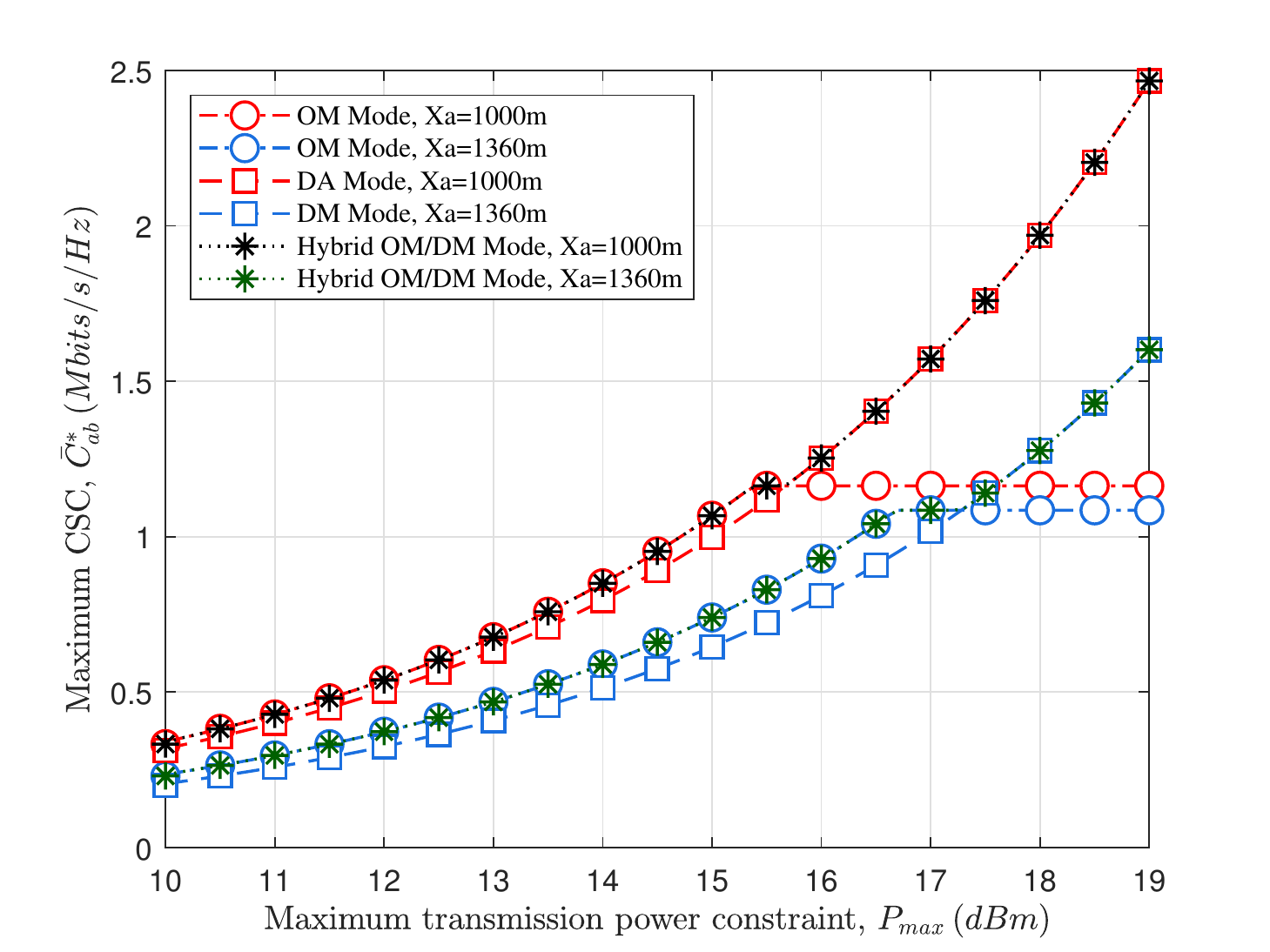} \label{fig62} 
		}
\vspace{-0.3cm}
    \caption{Covert performance vs. maximum transmission power constraint $P_{max}$ with $x_a \in\{1000m, 1360m\}$. }
    \label{fig6}
        \vspace{-0.5cm}
\end{figure}

To explore the impact of the maximum power constraint $P_{max}$ on the maximum ECR $\bar{R}_{ab}^{*}$ and CSC $\bar{C}_{ab}^{*}$ under both two modes, we summarize in Fig.~\ref{fig6} how they vary with $P_{max}$ with $x_a=1000$m and $x_a=1360$m, respectively. From Fig.~\ref{fig61}, we can see that as $P_{max}$ increases, $\bar{R}_{ab}^{*}$ gradually increases and then tends to a constant under the OM transmission mode. It can be explained as that a larger $P_{max}$ can allow Alice uses more power to transmit information resulting in a smaller $P_{out}$. But the covert performance constraint limits that $P_a$ cannot increase all the time. Thus, when $P_{max}$ is enough large, $\bar{R}_{ab}^{*}$ reaches its maximum value and keeps constant. For the DM transmission mode, due to the fast fading of mmWave, even given a large $P_{max}$, the received power at Willie is still small and the DEP is still high. Thus, Alice can still increase $P_a$ to improve the effective covert rate. The corresponding analysis can be applied to illustrate the phenomenon of the covert capacity $\bar{C}_{ab}^{*}$ in Fig.~\ref{fig62}, we omit here.

\begin{figure}[tb]
    \centering
		\subfigure[ECR $\bar{R}^{*}_{ab}$]{
		\includegraphics[width=0.225\textwidth]{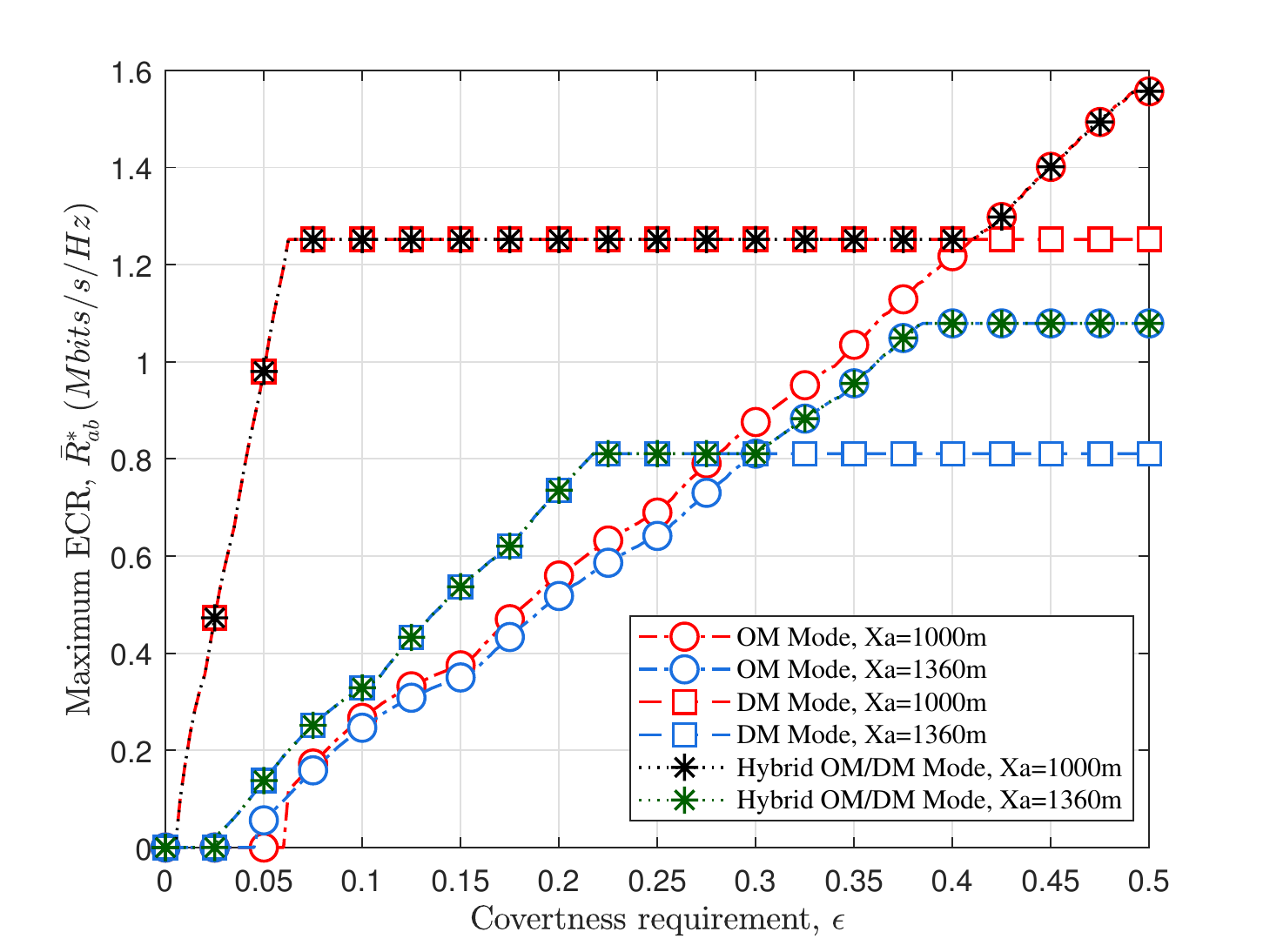}\label{fig:CR}
		}
		\subfigure[CSC $\bar{C}^{*}_{ab}$]{
		\includegraphics[width=0.225\textwidth]{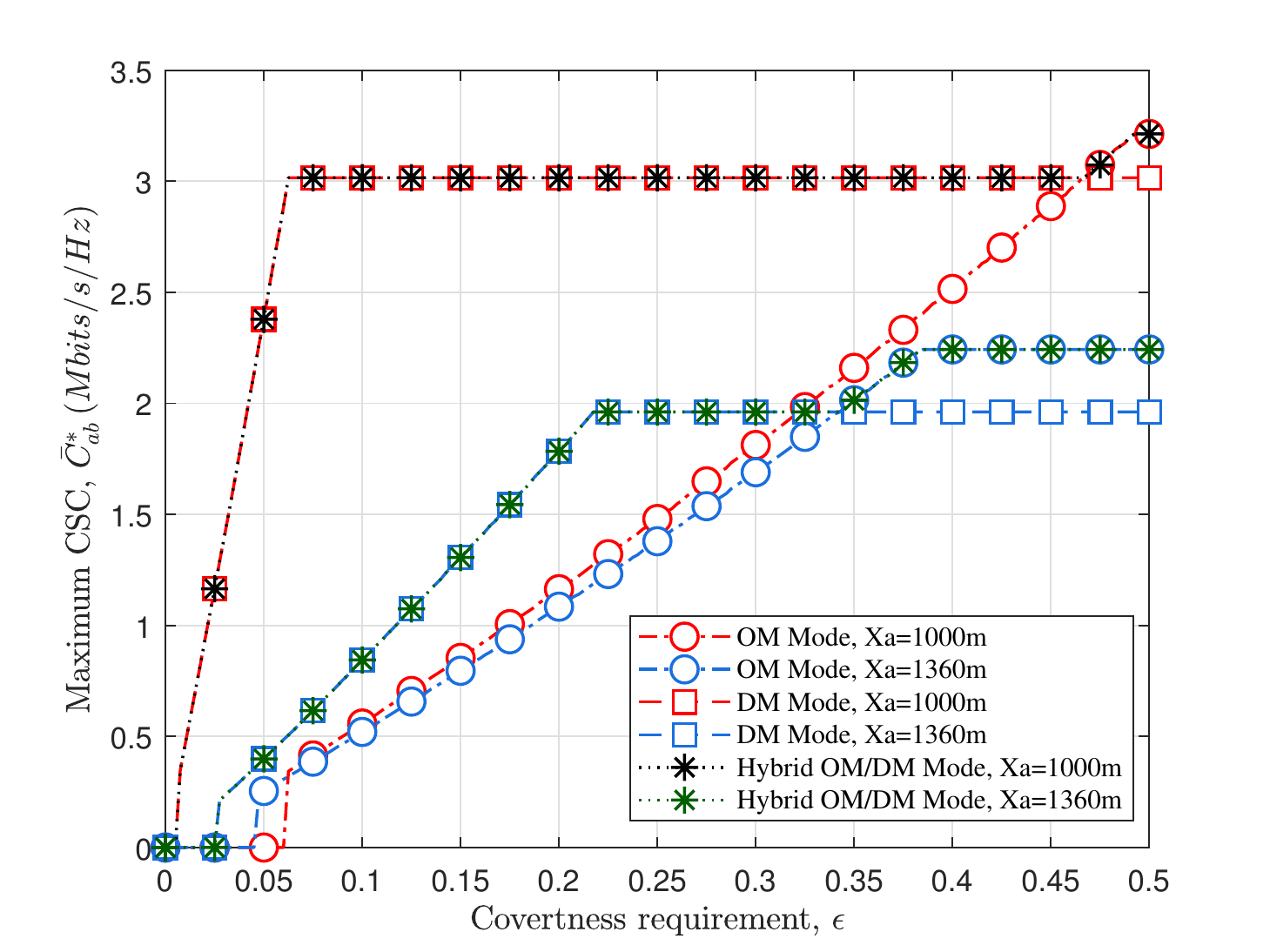}\label{fig:CC} 
		}
\vspace{-0.3cm}
    \caption{Covert performance vs. covertness constraint $\epsilon$ with $x_a \in\{1000m, 1360m\}$ and $P_{max}=20$dBm.}
    \label{fig:covert_epsilon}
    \vspace{-0.5cm}
\end{figure}

We then investigate the impact of the covertness constraint $\epsilon$ on the maximum effective covert rate $\bar{R}_{ab}^{*}$ and covert capacity $\bar{C}_{ab}^{*}$. The results are summarized in Fig.~\ref{fig:covert_epsilon}. We can observe from Fig.~\ref{fig:covert_epsilon} that as $\epsilon$ increases, both $\bar{R}^{*}_{ab}$ and $\bar{C}^{*}_{ab}$ first increase and then tends to remain unchanged. This can be explained as follows. As the covert constraints are gradually relaxed, Alice can adopt a large $P_a$ to transmit the covert information and thus both $\bar{R}_{ab}^{*}$ and $\bar{C}_{ab}^{*}$ increase. But when $P_a$ is large enough, the transmission power remains unchanged due to the maximum power constraint $P_{max}$. We further observe from Fig.~\ref{fig:CR} and Fig.~\ref{fig:CC} that under the hybrid mode, adopting the OM transmission mode can achieve better performance when the covertness restrictions are relaxed enough, while adopting the DM transmission mode is better when the covertness constraint is strict.

\begin{figure} [tb]
    \centering
    \includegraphics[width=0.3\textwidth]{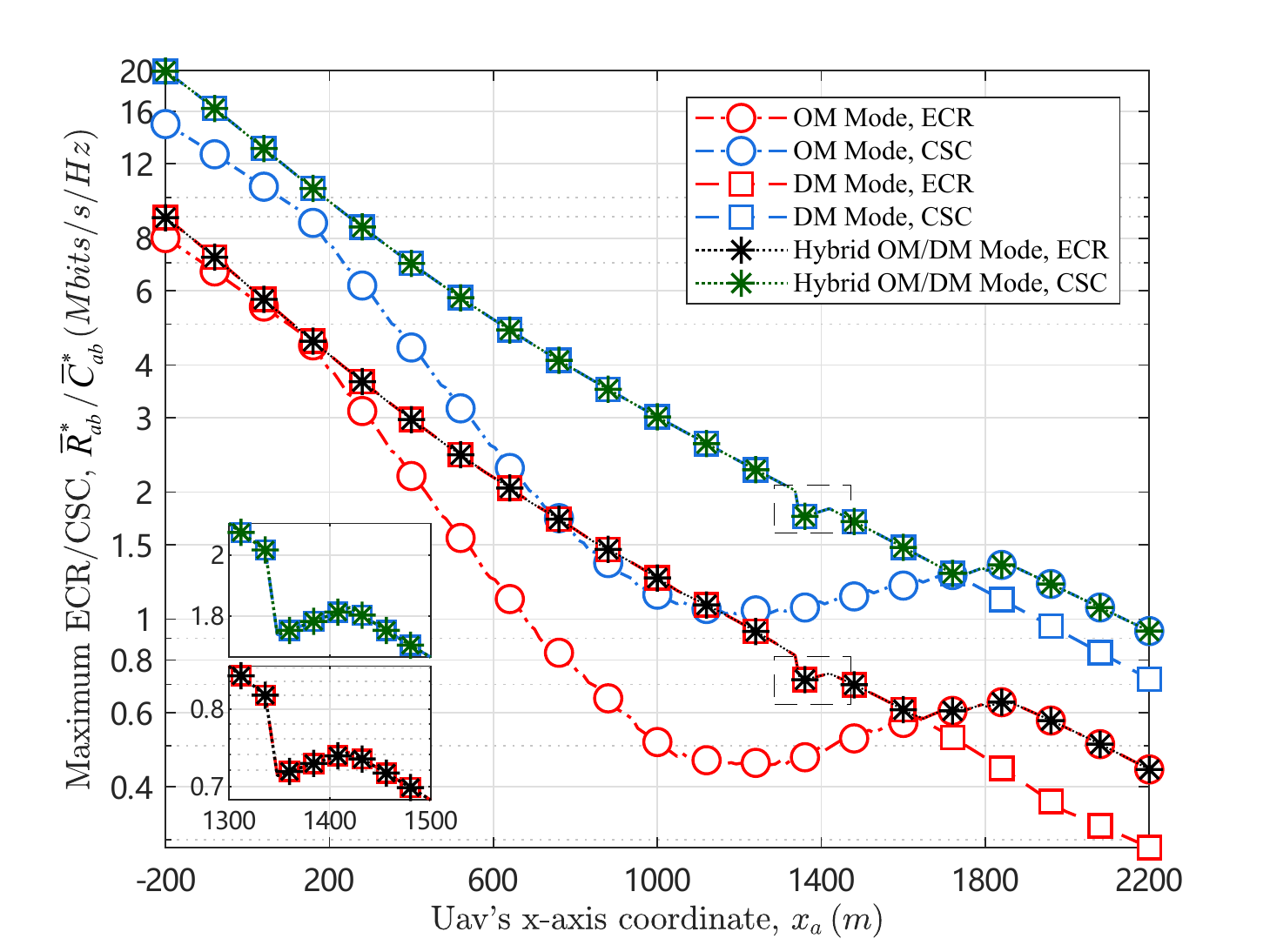}
    \vspace{-0.3cm}
    \caption{Covert performance vs. horizontal position of UAV $x_a$.} 
    \label{fig:covert_xa}
      \vspace{-0.5cm}
\end{figure}

Finally, we investigate the performance of different $x_a$ during the movement of Alice. We summarized in Fig.~\ref{fig:covert_xa} how the covert performance varies with $x_a$ for a setting of $\rho=2$dB and $P_{max}=20$dBm. We can see from Fig.~\ref{fig:covert_xa} that during Alice's movement from $(-200,0,500)$ to $(2200,0,500)$, the covert performance first decreases and then increases and finally decreases under both transmission modes. The interesting behavior can be explained as follows. When Alice moves away from Bob and closes to Willie, due to the negative effects brought by the path attenuation and covert constraint, both $R^{*}_{ab}$ and $\bar{C}^{*}_{ab}$ will decrease quickly. When Alice moves away from Willie, the power received by Willie decreases, and it can increase $P_a$ appropriately to improve the transmission efficiency. When Alice moves far away from Willie and Bob enough, Alice can adopt a larger transmission power satisfying the covertness constraint, however, it cannot offset the negative impact of large-scale attenuation on transmission performance. Thus, both $\bar{R}^{*}_{ab}$ and $\bar{C}^{*}_{ab}$ will decrease. Moreover, from Fig~\ref{fig:covert_xa}, we can see that when Alice is far away from Willie, the DM transmission mode outperforms the OM transmission mode in terms of both $\bar{R}^{*}_{ab}$ and $\bar{C}^{*}_{ab}$, but when Alice moves a certain distance away from Willie, the OM transmission mode is better. On the one hand, the power received by Bob is smaller due to the rapid attenuation of mmWave; on the other hand, Willie is in the main lobe radiation range such that Willie has good detection performance. By comparing the two modes, we can derive the optimal transmission mode as Fig.~\ref{fig:covert_xa} in the hybrid microwave/mmWave A2G Systems.


\section{Conclusion}\label{sec_conclusion}

This paper investigated the covert communication in a hybrid Microwave/mmWave A2G wireless communication system. Based on our theoretical performance analysis and covert performance optimization under both OM and DM transmission modes, we proposed a new hybrid OM/DM transmission mode for covert performance enhancement. The results in this paper revealed that the hybrid transmission mode can lead to a significant improvement of overt performance than the pure OM or DM mode in the concerned A2G system. It is expected that this work can provide meaningful insights into the covert communication scheme design in more general and more complicated UAV networks, e.g., UAV with multiple band antennas, UAV swarm, etc.

\begin{appendices}

\section{Proof of Lemma~\ref{lemma:calu2}} \label{nnap1}
Here, we provide the calculation steps for (\ref{Eq:calu2}) as follows
\begin{equation} \label{nneqa1}
\begin{aligned}
\int_{0}^{a}\!\! xf_{|h^{\!o,\mathbb{B}}_{ aw}|^{2}}(x)dx =xF_{|h^{o,\mathbb{B}}_{aw}|^{2}}(x)\textbf{\Big|}_{0}^{a}-\int_{0}^{a}\!\!F_{|h^{o,\mathbb{B}}_{aw}|^{2}}(x)dx.
\end{aligned}
\end{equation}
Similarly to (\ref{2eq81_NJ}), we can derive the first term on the right as
\begin{equation} \label{nnEq_ap1A}
\begin{aligned}
 xF_{|h^{ o,\mathbb{B}}_{aw}|^{2}}(x)\textbf{\Big|}_{0}^{a}=\!a\!\left(\!1\!\!-\!\exp{\!\left(\!-\left[2a (k_{\mathbb{B}}\!+\!1)\right]^\frac{\nu(\sqrt{2k_{\mathbb{B}}})}{2}e^{\mu(\sqrt{2k_{\mathbb{B}}})}\!\right)}\!\right).
\end{aligned}
\end{equation}
Then, we derive the second term on the right as
\begin{align} \label{nneqa2}
&\!\int_{0}^{a}\!\!\!\!F_{|h^{ o,\mathbb{B}}_{ aw}|^{2}}(x)dx\!=\!\!\int_{0}^{a}\!\!\!1\!-\!\exp{\left(\!-[2x(k_{\mathbb{B}}\!+\!1)]^\frac{\nu(\sqrt{2k_{\mathbb{B}}})}{2}e^{\mu(\sqrt{2k_{\mathbb{B}}})}\!\right)}dx \notag \\
&=a-\int_{0}^{a}\exp{\left(-[2x(k_{\mathbb{B}}+1)]^\frac{\nu(\sqrt{2k_{\mathbb{B}}})}{2}e^{\mu(\sqrt{2k_{\mathbb{B}}})}\right)}dx \notag \\
& \overset{(a)}{=}a-\frac{\gamma\left ( \frac{2}{\nu(\sqrt{2k_{\mathbb{B}}})}, [2a(k_{\mathbb{B}}+1)]^\frac{\nu(\sqrt{2k_{\mathbb{B}}})}{2} e^{\mu(\sqrt{2k_{\mathbb{B}}})}\right)}{(k_{\mathbb{B}}+1)\nu(\sqrt{2k_{\mathbb{B}}})e^{\frac{2\mu(\sqrt{2k_{\mathbb{B}}})}{\nu(\sqrt{2k_{\mathbb{B}}})}}},
\end{align}
where step (a) is according to $\int_{0}^{u}x^{m}e^{-bx^{n}}dx=\frac{\gamma(v,bu^{n})}{nb^{v}}, v=\frac{m+1}{n}$ as in \cite[Eq. (3.381.8)]{gradshteyn2014table}. Substituting (\ref{nnEq_ap1A}) and (\ref{nneqa2}) into (\ref{nneqa1}), we can obtain (\ref{Eq:calu2}).

\section{Proof of Lemma~\ref{lemma:calu}} \label{ap1}
Here, we provide the calculation steps for (\ref{Eq:calu}) as follows
\begin{equation} \label{eqa1}
\begin{aligned}
\int_{0}^{a}\!\!\!\! xf_{|h^{\! m,\mathbb{B}}_{\! aw}|^{2}}(x)dx \!=xF_{|h^{\! m,\mathbb{B}}_{\! aw}|^{2}}(x)\textbf{\Big|}_{0}^{a}\!-\!\int_{0}^{a}\!\!F_{|h^{\! m,\mathbb{B}}_{\! aw}|^{2}}\!(x)dx.
\end{aligned}
\end{equation}
According to (\ref{eq10}), we can derive the first term on the right as
\begin{equation} \label{Eq_ap1A}
\begin{aligned}
 xF_{|h^{ m,\mathbb{B}}_{aw}|^{2}}(x)\textbf{\Big|}_{0}^{a}=a\sum\nolimits_{r=0}^{S_{\mathbb{B} }}\binom{S_{\mathbb{B} }}{r}(-1)^{r}e^{-r\xi _{\mathbb{B}}a}.
\end{aligned}
\end{equation}
Similarly, the second term on the right can be derived as
\begin{align} \label{eqa2}
\int_{0}^{a}&F_{|h^{ m,\mathbb{B}}_{aw}|^{2}}(x)dx  =\int_{0}^{a}\sum_{r=0}^{S_{\mathbb{B} }}\binom{S_{\mathbb{B} }}{r}(-1)^{r}e^{-r\xi _{\mathbb{B}}x}dx \notag \\
& =\int_{0}^{a}\left[1+\sum\nolimits_{r=1}^{S_{\mathbb{B} }}\binom{S_{\mathbb{B} }}{r}(-1)^{r}e^{-r\xi _{\mathbb{B}}x}\right]dx\notag \\
&=a+\sum\nolimits_{r=1}^{S_{\mathbb{B} }}\binom{S_{\mathbb{B} }}{r}\frac{(-1)^{r}}{r\xi _{\mathbb{B}}}\left(1-e^{-r\xi _{\mathbb{B}}a}\right).
\end{align}
Then, substituting (\ref{Eq_ap1A}) and (\ref{eqa2}) into (\ref{eqa1}), we can obtain (\ref{Eq:calu}).

\end{appendices}

\bibliography{bibfile}
\bibliographystyle{IEEEtran}

\end{document}